\documentclass[a4paper, 11pt]{article}
\usepackage{amsmath}
\usepackage{amssymb,esint}
\usepackage{amscd}
\usepackage{xspace}
\usepackage{fancyhdr}
\usepackage{color}
\usepackage{srcltx}
\usepackage{slashbox}
\usepackage{makecell}
\usepackage{cite}
\usepackage{longtable}
\usepackage{float}
\usepackage{authblk}
\usepackage{graphicx}

\usepackage{enumerate}

\usepackage{mathrsfs}

\newenvironment{proof}[1][Proof]{\textbf{#1.} }{\hfill\rule{0.5em}{0.5em}}
{\catcode`\@=11\global\let\AddToReset=\@addtoreset
\AddToReset{equation}{section}

\newcommand{\Z}{{\mathbb Z}}
\newcommand{\N}{{\mathbb N}}
\newcommand{\F}{{\mathbb F}}

\newcommand{\C}{{\mathbb C}}

\newcommand{\cW}{{\mathcal W}}

\newcommand{\cG}{{\mathcal G}}

\newcommand{\eqccz}{\overset{\textrm{\tiny CCZ}}{\sim}}
\newcommand{\eqea}{\overset{\textrm{\tiny EA}}{\sim}}

\newcommand{\noteqea}{\overset{\textrm{\tiny EA}}{\not\sim}}

\newcommand{\clsEA}{\mathscr{C}_\textrm{EA}}
\newcommand{\clsCCZ}{\mathscr{C}_\textrm{CCZ}}
\newcommand{\Tr}{{\rm Tr}}
\newcommand{\Ker}{{\rm Ker}}
\newcommand{\Ima}{{\rm Im}}
\newcommand{\LK}{{\mathcal{LK}}}
\newcommand{\NL}{\textup{\texttt{NL}}}
\newcommand{\LL}{\textup{\texttt{L}}}
\newcommand{\caliL}{{\mathscr{L}}}
\newcommand{\eps}{\varepsilon}

\newcommand{\Fp}{\mathbb{F}_{p}}
\newcommand{\Fpn}{\mathbb{F}_{p^n}}
\newcommand{\Fpm}{\mathbb{F}_{p^m}}
\newcommand{\Fqn}{\mathbb{F}_{q^n}}

\newcommand{\Fbn}{\mathbb{F}_{2^n}}
\newcommand{\Fbm}{\mathbb{F}_{2^m}}

\newcommand{\Fq}{\mathbb{F}_{q}}

\newtheorem{theorem}{Theorem}[section]
\numberwithin{theorem}{section}
\newtheorem{lemma}[theorem]{Lemma}
\newtheorem{definition}[theorem]{Definition}
\newtheorem{corollary}[theorem]{Corollary}
\newtheorem{prop}[theorem]{Proposition}
\newtheorem{remark}[theorem]{Remark}
\newtheorem{example}[theorem]{Example}

\numberwithin{equation}{section}

\begin{document}
\title{On the Functions Which are CCZ-equivalent but not  EA-equivalent to Quadratic Functions over $\Fpn$}

 \author[1]{Jaeseong Jeong}
 \author[2]{Namhun Koo}
 \author[3]{Soonhak Kwon}
  %  {Namhun Koo \thanks{Institute of Mathematical Sciences, Ewha Womans University, Seoul, Korea, Email: nhkoo@ewha.ac.kr}},
   % {Soonhak Kwon }
   \affil[1]{\footnotesize Department of Innovation Center for Industrial Mathematics, NIMS, Seongnam, Republic of  Korea}
   \affil[2]{\footnotesize Institute of Mathematical Sciences, Ewha Womans University, Seoul, Republic of Korea}
    \affil[3]{\footnotesize Applied Algebra and Optimization Research Center, Sungkyunkwan University, Suwon, Republic of Korea}
   \date{}
\maketitle

\begin{abstract}
For a given function $F$ from  $\mathbb F_{p^n}$ to itself, determining whether there exists a function which is
CCZ-equivalent but EA-inequivalent to $F$ is a very important and interesting problem. For example, K\"olsch
\cite{KOL21} showed that there is no function which is CCZ-equivalent but EA-inequivalent to the inverse function. On
the other hand, for the cases of Gold function $F(x)=x^{2^i+1}$ and $F(x)=x^3+\Tr(x^9)$ over $\Fbn$, Budaghyan, Carlet
and Pott (respectively, Budaghyan, Carlet and Leander) \cite{BCP06, BCL09FFTA} found functions which are CCZ-equivalent
but EA-inequivalent to $F$. In this paper, when a given function $F$ has a component function which has a linear
structure, we present functions which are CCZ-equivalent to $F$, and if suitable  conditions are satisfied, the
constructed functions are shown to be EA-inequivalent to $F$. As a consequence, for every quadratic function $F$ on
$\mathbb F_{2^n}$ ($n\geq 4$) with nonlinearity $>0$ and differential uniformity $\leq 2^{n-3}$, we explicitly
construct functions which are CCZ-equivalent but EA-inequivalent to $F$. Also for every non-planar quadratic function
on $\Fpn$ $(p>2, n\geq 4)$ with $|\cW_F|\leq p^{n-1}$ and differential uniformity $\leq p^{n-3}$, we explicitly
construct functions which are CCZ-equivalent but EA-inequivalent to $F$. As an application, for a proper divisor $m$ of
$n$, we present many examples of $(n,m)$-functions $F$ on $\Fpn$ such that the CCZ-equivalence class of $F$ is strictly
larger that the EA-equivalence class of $F$.
\end{abstract}

\section{Introduction}

Vectorial Boolean functions are widely applied in cryptography and coding theory. Many researchers have been interested
in constructing vectorial Boolean functions with good cryptographic properties. For example, related to S-Box in AES,
\emph{The big APN problem}, finding an APN permutation on even dimension, is one of the most prominent problems that
many researchers are interested in. Moreover, functions with low differential uniformity, low boomerang uniformity, and
high nonlinearity have also garnered significant interest among researchers. Please see
\cite{Can16,Car20,CS17,Hou04,MMM22,Mes16} for surveys and general theories on these subjects.

When new cryptographic functions are proposed, researchers consider whether they are cryptographically equivalent to
other known functions or not. One of such equivalence relations is proposed by Carlet, Charpin and Zinoviev \cite{CCZ},
and it is referred to as CCZ-equivalence named after the proposers. Several cryptographic properties such as
differential spectrum and extended Walsh spectrum are invariant under CCZ-equivalence. To show that the given two
functions are not CCZ-equivalent, one may compare their differential spectrum or extended Walsh spectrum. However,
proving CCZ-equivalence of two functions with the same differential spectrum and extended Walsh spectrum is not easy in
general. The EA-equivalence (extended affine equivalence) is another cryptographic equivalence relation, and it is well
known that if two functions are EA-equivalent then they are CCZ-equivalent. It is also known that EA-equivalence
preserves the algebraic degree, while two CCZ-equivalent functions may have different algebraic degrees. Two very
important topics related with CCZ and EA equivalences are
\begin{enumerate}
 \item[1.] determining whether the two equivalences are the same notion for some class of functions.
 \item[2.] finding CCZ-equivalent but EA-inequivalent functions if CCZ-equivalence is strictly more general than
           EA-equivalence for some class of functions.
\end{enumerate}

\noindent For the first topic, it is known that CCZ-equivalence and EA-equivalence are the same notion
\begin{enumerate}
 \item[-] for all planar functions over $\Fpn$ \cite{BH08, KP08}.
 \item[-] for all Boolean functions over $\Fbn$ \cite{BC10eqccz}.
 \item[-] for all vectorial bent Boolean functions over $\Fbn$ \cite{BC09eqccz}.
 \item[-] for two quadratic APN functions over $ \Fbn$ \cite{Yo12eqccz}.
 \item[-] for two non-invertible power functions over $ \Fpn$ \cite{Dem18}.
\end{enumerate}

\noindent For a permutation $F$, its compositional inverse $F^{-1}$ is CCZ-equivalent to $F$. Therefore CCZ-equivalence
class of $F$ contains EA-equivalence classes of both $F$ and $F^{-1}$.
 For power functions over $\Fpn$, Dempwolff \cite{Dem18} showed that two power functions $x^s$ and $x^t$ are
CCZ-equivalent to each other  if and only if both $s$ and $t$  (or $s$ and $t^{-1} \pmod{p^n-1}$) are in the same
cyclotomic coset $\pmod{p^n-1}$. Recently, K\"olsch \cite{KOL21} showed that every function over $\Fbn$ ($n\geq 5$)
that is CCZ-equivalent to the inverse function $x^{2^n-2}$ is already EA-equivalent to it.

Our main interest in this paper is the second topic. That is, finding functions which are CCZ-equivalent but
EA-inequivalent to a given function, if that is possible.
 Since
both of CCZ and EA equivalences are not easy to show, there are few research results on this topic.
 Canteaut et al. \cite{CCP22} presented a new algorithm that efficiently solves the EA-recovery problem for quadratic functions.
  Kaleyski\cite{Kal22} studied invariants of EA-equivalence and their
  efficient computations.
In \cite{BCP06}, Budaghyan, Carlet and Pott presented functions with algebraic degrees $\geq 3$ which are
CCZ-equivalent to Gold functions over binary fields. Also,  Budaghyan, Carlet and Leander \cite{BCL09FFTA} presented
APN functions with degrees $\geq 3$ which are CCZ-equivalent to $F(x) = x^3+\Tr (x^9)$ over $\Fbn$. To the authors'
knowledge, these two functions ($x^{2^i+1}$ and $x^3+\Tr (x^9)$) are the only known  examples of quadratic functions
which have CCZ-equivalent but EA-inequivalent functions (which are not compositional inverses) with algebraic degrees
$\geq 3$. Recently, the relations between CCZ and EA equivalences are studied by Canteaut and Perrin \cite{CP19} by
introducing the notions of $t$-twisting and thickness spectrum,
 and also by Budaghyan, Calderini and Villa \cite{BCV20} by further analyzing the linear structure of given functions.

In this paper, by extending previous works \cite{BCP06, BCL09FFTA, CP19, BCV20}, we show that every quadratic $(n,n)$-
function over arbitrary finite field $\Fpn$ have CCZ-equivalent functions whose algebraic degrees are greater than $2$,
if both the linearity and the differential uniformity of the target quadratic function are not so bad. Our general
approach is to consider functions over $\Fpn$ that have a component function which has a linear structure. For those
functions, we construct a class of CCZ-equivalent functions, where the algebraic degrees are either $3$ or $4$.

 In
particular, for every quadratic $(n,n)$-function $F$ on $\mathbb F_{2^n}$ ($n\geq 4$) with nonlinearity $>0$ and
differential uniformity $\leq 2^{n-3}$, we explicitly construct functions which are CCZ-equivalent but EA-inequivalent
to $F$. Also, for every non-planar quadratic $(n,n)$-function $F$ on $\mathbb F_{p^n}$ ($p$ is odd and $n\geq 4$) with
linearity $\leq p^{n-1}$ and differential uniformity $\leq p^{n-3}$, we  construct functions which are CCZ-equivalent
but EA-inequivalent to $F$. As a straightforward generalization of our methods, we show that, for every quadratic
$(n,m)$-function $F$ on $\Fpn$ with $m|n$ and $4\leq m <n$ satisfying $|\cW_F|\leq p^{n-1}$ and $p^{n-m}< \Delta_F \leq
p^{n-3}$ where $p$ is an arbitrary prime, we show that the CCZ-equivalence class of $F$ is strictly larger than the
EA-equivalence class of $F$.

%Our proof is based on the fact that the algebraic degree of the resulting function
% is $3$ for the binary field, and is $3$ or $4$ for odd characteristic field.

The rest of this paper is organized as follows. In Section 2, we introduce some preliminaries which are necessary in
subsequent sections. In Section 3, we introduce several properties of linear structure, linear translator and linear
kernel. In Section 4, we mention the properties of Walsh transforms of strongly plateaued functions over $\Fpn$. In
Section 5, we construct CCZ-equivalent functions to the functions with a component function having a  linear structure.
In Section 6, considering quadratic functions with good linearity and differential uniformity, we show that the
algebraic degrees of the resulting functions are greater than two, and hence we have functions that are CCZ-equivalent
but EA-inequivalent to quadratic functions. In Section 7, we give concrete examples. In Section 8, we discuss
applications of our results to $(n,m)$-functions. In Section 9, we give concluding remarks.

\section{Preliminaries}\label{preli}

\noindent Let  $\F_{p^n}$ the finite field of $p^n$ elements where $n\in \N$ and $p$ is a prime number.   Let
$\F_{p^n}^\times$ be the multiplicative group consisting of all invertible elements in $\F_{p^n}$.
  To measure the resistance of the block cipher (especially a $S$-box inside of the block cipher) against
the differential cryptanalysis, for every $a\in \F_{p^n}^\times$ and $b\in \F_{p^n}$, we define the difference
distribution table $\Delta_F(a,b)$ and the differential uniformity $\Delta_F$ for a  vectorial $(n,n)$-function
$F:\F_{p^n}\to \F_{p^n}$ as follows
\begin{align*}
\Delta_F(a,b)=\#\{x\in \F_{p^n}: F(x+a)-F(x)=b\}, \\
\Delta_F =\max\{\Delta_F(a,b): a\in \F_{p^n}^\times, b\in \F_{p^n}\}.
\end{align*}
 We also
  define the \textit{differential} of $F$   to the direction of $a\in \F_{p^n}^\times$ by
$$D_{a}F(x)=F(x+a)-F(x) \mbox{ for all } x\in \F_{p^n} $$
so that $ \Delta_F(a,b)=\#\{x\in \F_{p^n}: D_{a}F(x)=b\}$. Then the lower the value of $\Delta_F$ is, the better $F$
resists a differential attack. We say that $F$ is planar or perfect nonlinear (PN) if $\Delta_F=1$, and almost perfect
nonlinear (APN) if $\Delta_F=2$.
%After the work of Nyberg,
%differential uniformity, PN functions, and APN functions have been studied intensively by numerous authors, see for
%details  in a recent monograph \cite{Car20}.

On the other hand,  there is one important equivalence notion of vectorial Boolean functions called
Carlet-Charpin-Zinoviev (CCZ) equivalence. This notion was introduced in \cite{CCZ}, and the term
 CCZ-equivalence was first used in \cite{BCP06}. This equivalence is broader than other known equivalences
 such as affine equivalence and extended affine equivalence (EA-equivalence).

\begin{definition}
Two functions $F, F': \F_{p^n}\to \F_{p^n}$ are said to be
 \begin{enumerate}
 \item[--] affine equivalent if $F' = A_1 \circ F \circ A_2$,
 where $A_1$ and $A_2$ are affine permutations on $\F_{p^n}$.
 \item[--] extended affine equivalent (EA-equivalent) if
 $F' = A_1 \circ F \circ A_2  + A_3$, where $A_1$ and $A_2$ are affine
 permutations on $\F_{p^n}$ and $A_3: \F_{p^n}\to \F_{p^n}$ is an affine function.
 \item[--] CCZ-equivalent if there
exists an affine permutation $\mathscr{L}$ on $\Fpn \times \F_{p^n}$ such that $\mathscr{L}(\cG_F)= \cG_{F'}$ where
$\cG_F$ is the graph of the function $F$, that is $\cG_F = \{(x,F(x)) :  x \in \Fpn\} \subset \Fpn \times \F_{p^n}$.
 \end{enumerate}
\end{definition}

\begin{remark}\label{rmkequv1}
 We briefly write $F \overset{\textrm{\tiny affine}}{\sim} F'$ if
$F$ and $F'$ are affine equivalent. In a similar way, we write $F \overset{\textrm{\tiny EA}}{\sim} F'$ (resp. $F
\overset{\textrm{\tiny CCZ}}{\sim} F'$) if $F$ and $F'$ are EA-equivalent (resp. CCZ-equivalent). Then, we have the
following.
$$
 F \overset{\textrm{\tiny affine}}{\sim} F'  \Rightarrow
   F \overset{\textrm{\tiny EA}}{\sim} F' \Rightarrow
   F \overset{\textrm{\tiny CCZ}}{\sim} F'
$$
\end{remark}

\begin{remark}\label{rmkequv2}
Since the above mentioned three relations are equivalence relations, we may naturally think of the equivalence classes
for each equivalence relations. Namely, we denote $\mathscr{C}_\textrm{affine}(F)$ the set consisting of all $F'$ such
that $F' \overset{\textrm{\tiny affine}}{\sim} F$. In a similar manner, we denote $\mathscr{C}_\textrm{EA}(F)$ (resp.
$\mathscr{C}_\textrm{CCZ}(F)$) the set consisting of all $F'$ such that $F' \overset{\textrm{\tiny EA}}{\sim} F$ (resp.
$F' \overset{\textrm{\tiny CCZ}}{\sim} F$).
\end{remark}

 \noindent
 It is well-known that the differential uniformity is invariant under CCZ-equivalence, and consequently
 under EA-equivalence and affine equivalence. Another important notion invariant under CCZ-equivalence is
the linearity, which is the maximum absolute value of all Walsh coefficients of a given function.
 To explain Walsh coefficient and Walsh transform,  we need a notion of trace function
 (a very special form of linear function).
Let $n\in \N$ and $m$ be a positive divisor of $n$ such that $n=ms$. We denote $\Tr^n_m(x)$ the trace function from
$\F_{p^n}$ to $\F_{p^m}$, that is
$$\Tr^n_m(x)=\sum_{j=0}^{s-1}x^{p^{jm}}= x+x^{p^m}+x^{p^{2m}}+\cdots +x^{p^{(s-1)m}}. $$
When $m=1$, we write $\Tr (x)$ instead of $\Tr^n_1(x)$ if there is no confusion of related fields. For $d|m$ and $m|n$,
it is trivial to check
 $\Tr_d^n(x)=\Tr_d^m\Tr_m^n(x)$ for all $x\in \Fpn$.

Let $\xi=e^{\frac{2\pi i}{p}} \in \mathbb C$ be a complex primitive $p^{th}$ root of unity. For $(n,1)$-function
$f:\F_{p^n}\to \F_p$, its \textit{Walsh-Hadamard transform} $\cW_f$ is the Fourier transform of the function
$\xi^{f(x)}$, that is,
$$\cW_f(a)=\sum_{x\in \F_{p^n}}\xi^{f(x)-\Tr(ax)} \mbox{ for every } a\in \F_{p^n}.$$
We define $|\cW_f|$ be the maximum absolute value of all $\cW_f(a)$ for $a\in \Fpn$. That is,
   \begin{align}
   |\cW_f|=\max_{a\in \Fpn} |\cW_f(a)|,
   \end{align}
   where $|\cW_f(a)|^2=\cW_f(a) \overline{\cW_f}(a)$ and  $\overline{\cW_f}(a)$ is the complex conjugate of $ \cW_f(a)$.
   Due to Parseval's identity, one always have $|\cW_f|^2\geq p^n$, and
    a $(n,1)$-function $f$ achieving the lower bound $|\cW_f|=p^\frac{n}{2}$ is called a (generalized or $p$-ary) bent function.
    Please refer the book of Mesnager \cite{Mes16} for more about bent functions.

 For  $(n,n)$-function $F:\F_{p^n}\to \F_{p^n}$, we define its \textit{Walsh transform} $\cW_F:\F_{p^n}\times \F_{p^n}\to \C$
 by
$$\cW_F(a,b)= \cW_{\Tr(bF)}(a)=\sum_{x\in \F_{p^n}}\xi^{\Tr(bF(x)-ax)} \mbox{ for every } a, b\in \F_{p^n},$$
where value $\cW_F(a,b)$ is called the Walsh coefficient of $F$, and $\Tr(bF)$ is called a
 component function of $F$. In a similar way, we
define
\begin{align}
 |\cW_F|\overset{\textrm{def}}{=}\max_{a\in \Fpn, b\in \Fpn^\times} |\cW_F(a,b)|
   = \max_{a\in \Fpn, b\in \Fpn^\times} |\cW_{\Tr(bF)}(a)| = \max_{b\in \Fpn^\times} |\cW_{\Tr(bF)}|.
\end{align}
The value $ |\cW_F|$ is called the linearity of $F$ and is written as $\LL(F)=|\cW_F|$. For a binary case (i.e.,
$p=2$), we also define $\NL(F)$ (the nonlinearity of $F$) as
$\NL(F)=2^{n-1}-\frac{1}{2}|\cW_F|=2^{n-1}-\frac{1}{2}\LL(F)$. For arbitrary characteristic $p$,
 it is well-known that $ |\cW_F|$ is invariant under CCZ-equivalence.

\medskip
There is one more notion (which we will discuss very often) so called {\it algebraic degree} which is invariant under
 EA-equivalence but not under CCZ-equivalence.
 Any $(n,n)$-function $F$ on $\Fpn$ is uniquely represented by a univariate polynomial $F(x) \in \Fpn[x]$
 of degree $< p^n$. That is,
 $$
  F(x)=\sum_{i=0}^{p^n-1} a_i x^i, \quad a_i\in \Fpn.
 $$
For any integer $0 \leq k < p^n$ with its $p$-ary expansion $k=\sum_{j=0}^{n-1} k_j p^j$, we define $p$-weight of $k$,
$w_p(k)$, as $w_p(k)=\sum_{j=0}^{n-1}k_j$. An algebraic degree of $F$, $d^\circ(F)$,   is the maximum value among all
$w_p(i)$ with $a_i\neq 0$;
$$
  d^\circ(F)= \max_{\substack{{0\leq i< p^n}\\{a_i\neq 0} }} w_p(i) \quad \text{ where }
   F(x)=\sum_{i=0}^{p^n-1} a_i x^i.
$$

\noindent An $(n,n)$-function $F$ with its algebraic degree {\it two} is called a quadratic function, and a quadratic
function with no affine term,
$$
  F(x)=\sum_{i\leq j} a_{i,j} x^{p^i+p^j}, \quad a_{i,j} \in \Fpn,
$$
is called a DO-polynomial (Dembowski-Ostrom polynomial). It is understood that the summation indices $i\leq j$
 becomes strict (i.e., $i<j$) if $p=2$.

\section{Linear Structure, Linear Translator and Linear Kernel}

Let $q=p^n$ be a power of a prime $p$ and let $f : \Fqn \rightarrow \Fq$ be a $(n,1)$-function on $\Fqn$.
\begin{definition}
Let $\gamma\in \Fqn$ and $b\in \Fq$. We say $\gamma$ is a $b$-linear structure of $f$ if
\begin{align*}
 f(x+\gamma)-f(x)=b  \quad \text{ for all } x\in \Fqn,
\end{align*}
and $\gamma$ is called a linear structure of $f$ if $\gamma$ is a $b$-linear structure of $f$ for some $b \in \Fq$.
 We also say $f$ has a (nontrivial) linear structure if there exists $\gamma\in \Fqn^\times$ such that $\gamma$ is a linear
 structure of $f$.
\end{definition}

\begin{definition}
Let $\gamma\in \Fqn$ and $b\in \Fq$. We say $\gamma$ is a $b$-linear translator of $f$ if
\begin{align*}
 f(x+u\gamma)-f(x)=ub  \quad \text{ for all } x\in \Fqn \text{ and } \text{ for all } u\in \Fq,
\end{align*}
 and  $\gamma$ is called a linear translator of $f$ if $\gamma$ is a $b$-linear translator of $f$ for some $b \in \Fq$.  We also
say $f$ has a (nontrivial) linear translator if there exists $\gamma\in \Fqn^\times$ such that $\gamma$ is a linear
 translator of $f$.
\end{definition}

\noindent If $\gamma$ is a $b$-linear translator, then $\gamma$ is also a $b$-linear structure. However,
 it turns out that, over a prime field $\Fp$, `linear
translator' and `linear structure' are the same notion as is shown below.
\begin{lemma}\label{eqtranslator}
Let $b\in \Fp$ and $\gamma\in \Fpn^\times$. Let $f : \Fpn \rightarrow \Fp$ be a $(n,1)$-function. Then $\gamma$ is a
$b$-linear structure of $f$ if and only if $\gamma $ is a $b$-linear translator of $f$.
\end{lemma}

\noindent
\begin{proof}
$(\Leftarrow)$ is clear by taking $u=1$.
 $(\Rightarrow)$ For any $u\in \Fpn$,
 \begin{align*}
 f(x+u\gamma)-f(x) =\sum_{k=1}^u f(x+k\gamma)-f(x+(k-1)\gamma)=\sum_{k=1}^u f(x_k+\gamma)-f(x_k)=\sum_{k=1}^u b=ub,
 \end{align*}
 where $x_k=x+(k-1)\gamma$.
\end{proof}

\begin{remark}
It is well-known (e.g. see \cite{Kyu11}) that the set
$$\LK_f\overset{\text{\rm def}}{=} \{\gamma \in \Fpn : \gamma
\text{ is a linear structure of } f\}$$
 is a vector subspace of $\Fpn$ over $\Fp$, which is briefly sketched below for
the case of component functions. The subspace $\LK_f$ is called the linear kernel of $f$.
\end{remark}

 %\begin{align}
 %\mathcal{LK}&=\mathcal{LK}_{Tr(\beta F)} \notag \\
 %&=\{\gamma\in \F_{p^n} :
 %\Tr(\beta F(x+\gamma)-bF(x))=\Tr(\beta F(\gamma)-bF(0)) \textrm{ for all } x\in \mathbb
  %F_{p^n}  \}, \notag \\
 %&=\{\gamma\in \F_{p^n} :
%\Tr(\beta L_\gamma(x))=0 \textrm{ for all } x\in \mathbb F_{p^n}  \}
% =\{\gamma \in \Fpn : \Ima L_\gamma \subset \langle \beta \rangle^{\perp} \} \label{LK}
%\end{align}

Let $F(x) \in \F_{p^n}[x]$ be an $(n,n)$-function and let
\begin{align*}
L_\gamma(x) &\overset{\text{def}}{=}  L_{\gamma,F}(x)= F(x+\gamma)-F(x)-(F(\gamma)-F(0))=D_\gamma F(x)- D_\gamma F(0)
 \end{align*}
 be the differential $D_{\gamma}F(x)$ with its constant term
 killed.  Let $\beta \in \Fpn^\times$ and let
 \begin{equation}\label{lkexplain}
 \begin{split}
 \mathcal{LK}&\overset{\text{def}}{=}\mathcal{LK}_{\Tr(\beta F)}\\
 &=\{\gamma\in \F_{p^n} :
\Tr(\beta D_\gamma F(x))=\,\text{constant}\, =\Tr(\beta D_\gamma F(0)) \textrm{ for all } x\in \mathbb
F_{p^n}  \}  \\
 &=\{\gamma\in \F_{p^n} :
\Tr(\beta L_\gamma(x))=0 \textrm{ for all } x\in \mathbb F_{p^n} \} \\
 &=\{\gamma \in \Fpn : \Ima L_\gamma \subset \langle \beta
\rangle^{\perp} \} =\{ \gamma\in \Fpn : \beta \in (\Ima L_\gamma)^\perp\}
\end{split}
\end{equation}
be the linear kernel of the component function $\Tr(\beta F)$,
 where $W^{\perp} \overset{\textrm{def}}{=} \{ x\in \Fpn : \Tr(x w)=0
           \text{ for all } w\in W   \}$ for a subset $W$ of $\Fpn$.

 \begin{remark}\label{lkremark}
 From the definition of the linear kernel, one has
\begin{center}
 $\gamma$ is a linear structure of $\Tr(\beta F)$ $\Leftrightarrow$ $\gamma \in \LK_{\Tr(\beta F)}$
  $\Leftrightarrow$ $\Tr(\beta L_\gamma(x))=0$ for all  $x\in \Fpn$
\end{center}
 \end{remark}

\noindent
 From
 \begin{align}
 L_{\gamma+\gamma'}(x)&=F(x+\gamma+\gamma')-F(x)-(F(\gamma+\gamma')-F(0)) \notag \\
 &=F(x+\gamma+\gamma')-F(x+\gamma)-(F(\gamma')-F(0)) +F(x+\gamma)-F(x)-(F(\gamma)-F(0)) \notag\\
                &\qquad  +F(\gamma')-F(0)  +F(\gamma)-F(0)  -(F(\gamma+\gamma')-F(0)) \notag\\
      &=L_{\gamma'}(x+\gamma)+L_\gamma(x)+F(\gamma')+F(\gamma)-F(\gamma+\gamma')-F(0) \notag\\
      &=L_{\gamma'}(x+\gamma)+L_\gamma(x)-L_\gamma(\gamma'), \label{L-pseudolinear}
 \end{align}
 one has
  $$
 \Tr(\beta L_{\gamma+\gamma'}(x))=\Tr(\beta  L_{\gamma'}(x+\gamma))+\Tr(\beta L_\gamma(x))
                -\Tr(\beta L_\gamma(\gamma')),
  $$
  which implies that $\LK=\LK_{\Tr(\beta F)}$ is indeed a vector (sub)space of $\Fpn$.
  (Please refer Remark \ref{lkremark}.)
Note that $L_\gamma(\gamma')=L_{\gamma'}(\gamma)$ for all $\gamma,\gamma'\in \Fpn$. Also note that
     $L_\gamma(x)$ is linear if $F$ is quadratic.
     Therefore Equation \eqref{L-pseudolinear} for quadratic $F$ is simplified as
 \begin{align}
 L_{\gamma+\gamma'}(x)=L_\gamma(x)+L_{\gamma'}(x) \label{L-linear}.
 \end{align}

     %then $\Ima L_\gamma$ becomes a subspace of $\Fpn$.

\begin{remark}
 \noindent
\begin{enumerate}
 \item[1.] For a binary case, it is proven in \cite{CK08} that, up to the cyclotomic equivalence of the exponent,
 Gold function $x^{2^i+1}$ is the only (nonlinear) power function $x^d$ whose component function $\Tr(\beta x^{d})$
 has a linear structure for some $\beta$.
 \item[2.] It is almost trivial that PN function $F$ on $\Fpn$ ($p\neq 2$) does not have any (nonzero)
 component function which has a linear structure,
 since $\Delta_F=1$ implies that $D_\gamma F$ is surjective for all $\gamma\in \Fpn$. Therefore
  $D_\gamma \Tr(\beta F(x))=\Tr(\beta D_\gamma F(x))$
 cannot be a constant function for $\beta\neq 0$.
\end{enumerate}
\end{remark}

\begin{definition}\label{pbentdef}
Let $f: \Fpn \rightarrow \Fp$ be an $(n,1)$-function. We say $f$ is ($p$-ary) partially-bent if
 $D_\gamma f$ is either a constant or  a balanced function for every $\gamma \in \Fpn$.
  A vectorial $(n,n)$-function $F$ on $\Fpn$ is called
 strongly plateaued if $\Tr(bF)$ is partially-bent for every $b\in \Fpn$.
\end{definition}

\begin{remark}\label{splateaued}
The above definition of `strongly plateaued' comes from Definition $67$ in \cite{Car20}, where binary case was dealt.
For $p$-ary $(n,n)$-function, the above definition also make sense since Proposition \ref{pbentwalsh} implies that
every nonzero component function has one single nonzero absolute Walsh value.
\end{remark}

\begin{remark}\label{pbentremark}
 \noindent
 \begin{enumerate}
 \item[1.] Note that $p$-ary $(n,1)$-function $f$ is called balanced if
 $f^{-1}(i)\overset{\text{def}}{=}\{x\in \Fpn : f(x)=i\}$ is a subset of size $p^{n-1}$ for every $0 \leq i<p $.
 \item[2.] If $(n,n)$-function $F$ on $\Fpn$ is quadratic, then since $L_\gamma$ is linear,
  $\Tr(bL_{\gamma}(x))$ is either zero function or balanced function. That is, every quadratic $(n,n)$-function is
  strongly plateaued.
 \end{enumerate}
\end{remark}

\begin{prop}\label{pt-partially-bent}
 Properties of strongly plateaued $(n,n)$-functions varies depending on the characteristic $p$ as follows:
\begin{enumerate}
 \item[1.]
    When $p=2$, every strongly plateaued $(n,n)$-function $F$ on $\Fbn$ has a (nonzero) component function which has
    a (nontrivial) linear structure.
 \item[2.]
    When $p\neq 2$, a strongly plateaued $(n,n)$-function $F$ on $\Fpn$ has a (nonzero) component function which has
    a (nontrivial) linear structure
     if and only if $F$ is not PN(Perfect Nonlinear).
\end{enumerate}
\end{prop}

\noindent
\begin{proof}

 \noindent
  {\it 1.}  If $F$ has no component function $\Tr(\beta F)$ which has a linear structure except for $\beta=0$,
     then $D_\gamma(\Tr(\beta F))$ is balanced for all $\beta, \gamma \in \Fbn^\times$, which implies that
     $\Tr(\beta F)$ is bent for all $\beta\in \Fbn^\times$. This is impossible because there is no bent $(n,n)$-function
     for $p=2$ (See \cite{Nyb92}).

 \smallskip
   \noindent
  {\it 2.} If $\gamma$ is a linear structure of  $\Tr(\beta F)$ for some $\beta,\gamma\in \Fpn^\times$, then
   one has $\Tr(\beta L_\gamma(x))=0$ for all $x\in \Fpn$, i.e., $\beta \in (\Ima L_\gamma)^\perp$. Therefore $D_\gamma F$
   is not surjective. The other direction is also clear.
\end{proof}

\bigskip

\noindent As a special case of strongly plateaued functions, we have a more detailed description on the linear
structures of component functions of quadratic function $F$ below.
\begin{prop}\label{LS-quad}
Let $F$ be a quadratic $(n,n)$-function on $\Fpn$ such that $\Delta_F>1$. Then there exist $\gamma, \beta  \in
\Fpn^\times$ such that $\gamma$ is a linear structure of $\Tr(\beta F(x))$ (i.e., $\Tr(\beta L_\gamma(x))=0$ for all
$x\in \Fpn$). More precisely,
 \begin{enumerate}
 \item[1] If $p=2$, then for every $\gamma\in \Fpn^\times$, there exists $\beta\in \Fpn^\times$ such that
   $\Tr(\beta L_\gamma(x))=0$ for all $x\in \Fpn$.
 \item[2] If $p>2$ and $\Delta_F>1$ (i.e., $F$ is not PN), then there exist $\gamma, \beta\in \Fpn^\times$ such that
   $\Tr(\beta L_\gamma(x))=0$ for all $x\in \Fpn$.
 \end{enumerate}
\end{prop}

\noindent
\begin{proof}
\newline
\noindent {\it 1.} Since $L_\gamma(x)=F(x+\gamma)+F(x)+F(\gamma)+F(0)$ is linear such that
$L_\gamma(0)=0=L_\gamma(\gamma)$, one has $|\Ker L_\gamma|\geq 2$. Therefore $\Ima L_\gamma$ is a subspace of $\Fpn$ of
 dimension $\leq n-1$, that is, $L_\gamma$ is not surjective.
 Consequently there exists nonzero
 $\beta \in (\Ima L_\gamma)^\perp =\{b\in \Fpn : \Tr(b L_\gamma(x))=0 \text{ for all } x\in \Fpn \}$.

 \noindent
 {\it 2.} Since $F$ is not a PN function, one has $\Delta_F=p^s$ for some $1\leq s\leq n$.
  Therefore there exists $\gamma\in \Fpn^\times$ such that $|\Ker L_\gamma |= \Delta_F=p^s$.
  For this choice of $\gamma$, one has a subspace $\Ima L_\gamma$ of dimension $ n-s\leq n-1$. In the same way,
  there exists nonzero $\beta\in (\Ima L_\gamma)^\perp$.
\end{proof}

\section{Walsh Transform of ($p$-ary) Strongly Plateaued Functions over Arbitrary Characteristic}

The subsequent survey also applies to quadratic functions as a special case. Moreover, we will give concrete examples
 of  $F'$ satisfying $F\eqccz F'$ but $F\noteqea F'$
 only for quadratic $F$ in Section \ref{sec-concrete}.
  For a binary case, Walsh transform of strongly plateaued functions
   is well-documented (See for example \cite{CC03}). In fact, The case $p\neq 2$ is not so much
different from the case $p=2$. However, for the clarity of the exposition for general $p$,
 we briefly survey Walsh transform of strongly plateaued functions on $\Fpn$ with arbitrary
characteristic $p$.

Let $F:\F_{p^n}\to \F_{p^n}$ be a strongly plateaued function (or a quadratic function). Let $b\in \Fpn^\times$ and let
 $\mathcal{LK}=\mathcal{LK}_{\Tr(b F)}=\{\gamma\in \F_{p^n} :
\Tr(bL_\gamma(x))=0 \textrm{ for all } x\in \mathbb F_{p^n}  \}$ be the linear kernel of $\Tr(bF)$ with
$L_\gamma(x)=F(x+\gamma)-F(x)-(F(\gamma)-F(0))$
 (See Equation \eqref{lkexplain}).  Noting that $\bar{\xi}=\xi^{-1}$,
\begin{align*}
 |\cW_F(a,b)|^2 &=\sum_{y\in \F_{p^n}}\xi^{\Tr(bF(y)-ay)} \cdot
   \sum_{x\in \F_{p^n}}(\xi^{-1})^{\Tr(bF(x)-ax)}  \\
   &= \sum_{x,y\in \F_{p^n}}\xi^{-\Tr(a(y-x))} \xi^{\Tr(b[F(y)-F(x)])} \\
   &=  \sum_{x,\gamma\in \F_{p^n}}\xi^{-\Tr(a\gamma)}
   \xi^{\Tr(b[F(x+\gamma)-F(x)])}\\
 &=  \sum_{\gamma\in \F_{p^n}}\xi^{-\Tr(a\gamma)} \sum_{x\in \F_{p^n}}
   \xi^{\Tr(b[F(x+\gamma)-F(x)])}\\
   &=  \sum_{\gamma\in \F_{p^n}}\xi^{-\Tr(a\gamma)} \sum_{x\in \F_{p^n}}
   \xi^{\Tr(bL_\gamma(x))+\epsilon_\gamma}  \qquad (\epsilon_\gamma=\Tr(b[F(\gamma)-F(0)])) \\
    &=  \sum_{\gamma \in \mathcal{LK}}\xi^{-\Tr(a\gamma)} \sum_{x\in
    \F_{p^n}}
   \xi^{\epsilon_\gamma} + \sum_{\gamma \notin \mathcal{LK}}\xi^{-\Tr(a\gamma)} \sum_{x\in \F_{p^n}}
   \xi^{\Tr(bL_\gamma(x))+\epsilon_\gamma}.
\end{align*}
For $\gamma \notin \mathcal{LK}$, since $\Tr(bL_\gamma(x))$ is a balanced function if $F$ is strongly plateaued (or
since $\Tr(bL_\gamma(x))$ is nonzero linear function if $f$ is quadratic), $\Tr(bL_\gamma(x))$ takes every $u \in
\mathbb F_p$ equally many times. Therefore since $\sum_{u=0}^{p-1} \xi^u=0$, one gets
\begin{align*}
  |\cW_F(a,b)|^2  =  p^n\sum_{\gamma \in \mathcal{LK}}\xi^{-\Tr(a\gamma)+\epsilon_\gamma}.
\end{align*}

Now let $ \{\gamma_1, \gamma_2, \cdots, \gamma_k  \}$ be a basis for $\mathcal{LK}$ over $\mathbb F_p$ and let
$\epsilon_i \overset{\textrm{def}}{=}\epsilon_{\gamma_i}=\Tr(b[F(\gamma_i)-F(0)])$ for $1\leq i\leq k$. Then using the
vector space property of $\mathcal{LK}$, for any $\gamma=\sum_{i=1}^k e_i \gamma_i \in \mathcal{LK}$ with $e_i \in
\mathbb F_p$, one can express $\epsilon_\gamma$  as
$$
\epsilon_\gamma=\epsilon_{\sum_{i=1}^k e_i\gamma_i}=\sum_{i=1}^k e_i \epsilon_{\gamma_i}=\sum_{i=1}^k e_i \epsilon_i
\,\, \in \Fp.
$$
Extend $ \{\gamma_1, \gamma_2, \cdots, \gamma_k \}$ to a basis
 $ \{\gamma_1, \gamma_2, \cdots, \gamma_k, \gamma_{k+1},\cdots, \gamma_{n}
 \}$ for $\mathbb F_{p^n}$ and let $\{\alpha_1, \cdots,  \alpha_n
 \}$ be a unique dual basis of $\{\gamma_1, \cdots, \gamma_n\}$, i.e., $\Tr(\alpha_i\gamma_j)=1$ if $i=j$ and
  $\Tr(\alpha_i\gamma_j)=0$
 if $i\neq j$.

Let $a=\sum_{i=1}^n a_i \alpha_i \in \Fpn$ ($a_i\in \Fp$). Then
\begin{align*}
  |\cW_F(a,b)|^2  &=  p^n\sum_{\gamma \in \mathcal{LK}}
  \xi^{-\Tr\left((\sum_{i=1}^n a_i\alpha_i) (\sum_{j=1}^k
  e_j\gamma_j)\right)+\epsilon_\gamma}\\
  &=  p^n\sum_{\gamma \in \mathcal{LK}}
  \xi^{-(\sum_{i=1}^k a_ie_i)+\epsilon_\gamma} =
   p^n\sum_{\gamma \in \mathcal{LK}}
  \xi^{\sum_{i=1}^k (\epsilon_i-a_i)e_i}.
\end{align*}
As   $\gamma=\sum_{i=1}^k e_i \gamma_i$ runs through all elements of $\mathcal{LK}$, the vectors
 $\langle e_1, \cdots, e_k  \rangle$ runs through all vectors of
 $\mathbb F_p^k$. Therefore, if $\langle \epsilon_1-a_1, \cdots, \epsilon_k-a_k  \rangle$ is a nonzero vector,
 $\sum_{i=1}^k (\epsilon_i-a_i)e_i \in \mathbb F_p$ is balanced
 for all inputs $\langle e_1, \cdots, e_k  \rangle$. With these observations, one has the following well-known
 result of Walsh transform of strongly plateaued (or quadratic) functions over arbitrary characteristic $p$.
 %Consequently, letting
 %$\epsilon_i =\Tr(b[F(\gamma_i)-F(0)])$ $(1\leq i\leq k)$ and $a=\sum_{i=1}^n a_i\alpha_i$, one has the
 %following well-known result.

\begin{prop}{(Walsh transform of strongly plateaued functions)}\label{pbentwalsh}
 Let $F(x)\in \Fpn[x]$ be an $(n,n)$-function such that $F$ is strongly plateaued (or quadratic).
 Let $a, b\in \Fpn$. let $\LK=\LK_{\Tr(bF)}$ be the linear kernel of $\Tr(bF)$ of dimension $0\leq k\leq n$
  such that $\LK$ has a basis $\{\gamma_1,\cdots, \gamma_k\}$, and $\{ \gamma_1,\cdots,\gamma_k,\gamma_{k+1},\cdots \gamma_n  \}$
  is a basis for $\Fpn$. Let $\{\alpha_1,\cdots, \alpha_n\}$
   be a unique dual basis of $\{\gamma_1,\cdots, \gamma_n\}$.
   Let $a =\sum_{i=1}^n a_i\alpha_i$ and  $\epsilon_i =\Tr(b[F(\gamma_i)-F(0)])$ $(1\leq i\leq k)$. Then one has
\begin{equation*}\label{plateau_walsh}
 |\cW_F(a,b)|^2  = p^n\sum_{\gamma(=\sum_{j=1}^k e_j\gamma_j) \in \mathcal{LK}}
  \xi^{\sum_{i=1}^k (\epsilon_i-a_i)e_i}= \begin{cases}
p^{n+k} &\text{ if }  \epsilon_i=a_i \text{ for all } 1\leq i\leq k \\
0 &\text{ if not}.
\end{cases}
\end{equation*}
Consequently one has,
\begin{align*}
 |\cW_{\Tr(bF)}| &= p^\frac{n+k}{2}  \quad \text{ where } k= \dim \LK_{\Tr(bF)}, \\
 |\cW_{F}| &= p^\frac{n+k_m}{2} \quad \text{ where } k_m= \max_{b\in \Fpn^\times} \dim \LK_{\Tr(bF)}.
\end{align*}
\end{prop}

\section{CCZ-Equivalence of Some Functions Having Linear Structure}\label{sec-ccz}

\begin{lemma}\label{lemma-cczmap}
Let $\alpha, \beta \in \Fpn$ and $\gamma\in \Fpn^\times$. Let $\caliL : \Fpn \times \Fpn \rightarrow \Fpn \times \Fpn$
be defined by
\begin{align*}
 \mathscr{L} : \Fpn \times \Fpn &\rightarrow \Fpn \times \Fpn \\
(x,y) &\mapsto (x+\gamma\Tr(\alpha x+\beta y), y).
\end{align*}
Then, $\caliL$ is a linear permutation on $\Fpn \times \Fpn$ if and only if $\Tr(\alpha \gamma)\neq -1$.
\end{lemma}

\noindent
\begin{proof}
Since the linearity of $\caliL$ is clear, it is enough to show
$$
\caliL \text{ is one to one } \Longleftrightarrow \Tr(\alpha\gamma)\neq -1
$$
Now,
\begin{align*}
\caliL(x,y)=\caliL(x',y') &\Leftrightarrow (x+\gamma\Tr(\alpha x+\beta y), y)=(x'+\gamma\Tr(\alpha x'+\beta y'), y')\\
      &\Leftrightarrow y'=y \text{ and } x+\gamma\Tr(\alpha x)=x'+\gamma\Tr(\alpha x')\\
       &\Leftrightarrow y'=y \text{ and } x-x'= -\gamma\Tr(\alpha (x-x'))=-\gamma u \quad (u=\Tr(\alpha(x-x'))\in \Fp) \\
      &\Leftrightarrow y'=y \text{ and } -\gamma u = \gamma u\Tr(\alpha\gamma) \\
      &\Leftrightarrow y'=y \text{ and }  u(\Tr(\alpha\gamma)+1)=0 \quad (\because \gamma\neq 0) \\
      &\Leftrightarrow y'=y \text{ and }  x'=x \quad \text{or} \quad y'=y  \text{ and } \Tr(\alpha\gamma)=-1 \quad (x-x'=-\gamma u)
\end{align*}
\end{proof}

\begin{prop}\label{key-prop}
  Let $p$ be a prime. Suppose that $F_0(x) \in \Fpn[x]$ is an $(n,n)$-function such
  that $\gamma$ is a linear structure of $\Tr(\beta F_0)$
  for some $\beta, \gamma \in \Fpn^\times$. Then, choosing $\alpha, c \in \Fpn$ satisfying
  \begin{align*}
     \Tr(\alpha\gamma)=-2, \qquad  \Tr(c\beta\gamma)= -\Tr(\beta D_\gamma F_0(0)),
  \end{align*}
  and letting $F(x)=F_0(x)+cx$, the function
  \begin{align*}
 H(x)=x+\gamma\Tr\left(\alpha x +\beta F(x)\right)
  \end{align*}
  is an involution, i.e. $H\circ H(x)=x$ for all $x\in \Fpn$.
\end{prop}

\noindent
\begin{proof}
Since
 \begin{align*}
 L_{\gamma, F}(x)&= F(x+\gamma)-F(x)-(F(\gamma)-F(0)) \\
  &= F_0(x+\gamma)-F_0(x)+c\gamma -(F_0(\gamma)+c\gamma -F_0(0))=L_{\gamma,F_0}(x),
 \end{align*}
  $L_\gamma(x)=L_{\gamma,F}(x)=L_{\gamma,F_0}(x)$ is independent of choices of  $F(x)=F_0(x)+cx$ for any
  $c\in \Fpn$.
  On the other hand, $D_\gamma F(x)=D_\gamma F_0(x)+c\gamma$ so that
  \begin{align*}
 \Tr(\beta D_\gamma F(x))=\Tr\left(\beta (D_\gamma F_0(x)+c\gamma)\right)= \Tr(\beta D_\gamma F_0(x)) +\Tr(c\beta\gamma).
  \end{align*}
  Since $\gamma$ is a  linear structure of $\Tr(\beta F_0(x))$
  (i.e., $\Tr(\beta D_\gamma F_0(x))= \Tr(\beta D_\gamma F_0(0))$ for all $x\in \Fpn$),
   choosing any $c\in \Fpn$ satisfying
   \begin{align*}
  \Tr(c\beta\gamma)= -\Tr(\beta D_\gamma F_0(0))=-\Tr(\beta [F_0(\gamma)-F_0(0)]),
  \end{align*}
  one finds that $\gamma$ is a $0$-linear structure of $\Tr(\beta F(x))$ with $F(x)=F_0(x)+cx$. That is,
  \begin{align}\label{maineq1}
  \Tr(\beta D_\gamma F(x))=0 \text{ for all } x\in\Fpn.
  \end{align}

 \noindent
  Define $\eps$ as an $(n,1)$-function
  \begin{align*}
  \eps=\eps(x)=\Tr\left(\alpha x +\beta F(x)\right)
  \end{align*}
so that $H(x)=x+\gamma\eps$. Then
 \begin{align}
  F\circ H(x)=F(x+\gamma\eps)=F(x)+F(x+\gamma\eps)-F(x)=F(x)+D_{\gamma\eps}F(x), \label{fhfirst}
  \end{align}
  where  $ D_{\gamma\eps}F(x)$ is  defined as
 \begin{align*}
 D_{\gamma\eps}F(x) = F(x+\gamma\eps)-F(x)= F\left(x+\gamma \Tr(\alpha x +\beta F(x))\right)-F(x).
 \end{align*}

   \noindent
   Then
  \begin{align}
 H\circ H(x)&=H(x)+\gamma\Tr\left(\alpha H(x)+\beta F(H(x))  \right) \notag \\
        &=x+\gamma\eps +\gamma\Tr\left( \alpha(x+\gamma\eps)+\beta [F(x)+D_{\gamma\eps} F(x)]  \right) \label{fhsecond} \\
        &=x+\gamma\eps +\gamma\eps +\gamma\Tr\left(\alpha\gamma\eps +\beta D_{\gamma\eps} F(x)  \right) \notag \\
   &=x+\gamma\eps (2+\Tr(\alpha\gamma)) +\gamma\Tr\left(\beta D_{\gamma\eps} F(x)  \right) \notag \\
   &=x+\gamma\Tr\left(\beta D_{\gamma\eps} F(x)  \right).  \qquad (\because \Tr(\alpha\gamma)=-2) \label{maineq2}
  \end{align}
 Also, from Equation \eqref{maineq1}, one has $\Tr(\beta D_\gamma F(x))=0$  for all $x\in\Fpn$.
 Moreover, from Lemma \ref{eqtranslator} by taking $f(x)=\Tr(\beta F(x))$, one has
 \begin{align*}
 \Tr\left(\beta F(x+\gamma u) -\beta F(x)\right)
 =\Tr(\beta D_{\gamma u} F(x))=0
 \end{align*}
 for all $x\in \Fpn$ and for all $ u\in \Fp$. Therefore, from Equation (\ref{maineq2}), one has
  $H\circ H (x)=x$ for all $x\in \Fpn$.
  \end{proof}

\begin{theorem}\label{thm1}
Under the same conditions as in Proposition \ref{key-prop}, $F_0$ and $F\circ H$ are CCZ-equivalent. More precisely, we
have the following equivalence
$$
   F_0 \overset{\textrm{\tiny EA}}{\sim} F \overset{\textrm{\tiny CCZ}}{\sim} F\circ H.
$$
\end{theorem}
\begin{proof}
 $F_0 \overset{\textrm{\tiny EA}}{\sim} F$ is clear because $F(x)=F_0(x)+cx$.
  From Lemma \ref{lemma-cczmap}, one gets
  \begin{align*}
 \mathscr{L}(x, F(x))&=(x+\gamma\Tr(\alpha x +\beta F(x)), F(x))=(H(x),F(x)) \\
            &=(H(x), F\circ H\circ H(x))  \quad \because H \text{ is an involution. }\\
            &=(y, F\circ H(y))   \qquad (y=H(x))
  \end{align*}
  Since $\mathscr{L}$ is a permutation on $\Fpn\times\Fpn$ if $\Tr(\alpha\gamma)\neq -1$, and since
  $H$ is a permutation on $\Fpn$, the mapping $\mathscr{L}$ gives a graph isomorphism between
  the two sets $\{(x, F(x)) : x\in \Fpn\}$ and $\{(x, F\circ H(x)) : x\in \Fpn\}$.
\end{proof}

\begin{remark}
The resulting function $F\circ H$ is essentially, as you can see the expression in Proposition \ref{key-prop2}, a
generalized form arising from switching technique. This idea is originally due to Dillon and developed in many articles
such as \cite{EP09, LRS22}.
\end{remark}

\begin{remark}
It should be mentioned that, for a binary field, similar observations and results are found in \cite{CP19, BCV20}. In
\cite{CP19}, Canteaut and Perrin used the notion of function twisting and linear structure to derive a similar result
as above. It is also mentioned in page $242$ that ``\ldots it is very unlikely that the resulting function has degree
$2$.". Also in \cite{BCV20},  Budaghyan, Calderini and Villa obtained a similar construction using a linear structure.
In page $15$, they also remarked that ``\ldots such a transformation could lead to a function of degree $3$.". However,
none of the papers \cite{CP19, BCV20} give a new explicit example of cubic function which is CCZ-equivalent to given
quadratic function. At this moment, to the authors' knowledge, the following two functions

\smallskip
 \quad -- $x^{2^i+1}$ in \cite{BCP06},

 \quad -- $x^3+a^{-1}\Tr (a^3x^9)$ in \cite{BCL09FFTA},
\smallskip

\noindent are the only explicit examples of quadratic functions which have (non-obvious) CCZ-equivalent functions of
algebraic degree $\geq 3$.
\end{remark}

\medskip
 As a result of Theorem \ref{thm1}, it is natural to ask whether $F$ and $F\circ H$ is EA-equivalent or not. A straightforward approach is to
 examine the algebraic degree of $F\circ H$. We expect that $d^\circ (F) < d^\circ (F\circ H)$ under mild restrictions
 on the differential uniformity and the (non)linearity of $F$, but it seems to be a complicated problem for general situations.
  However, for a quadratic case (i.e., $d^\circ(F)=2$), we have a satisfactory result saying that
 \begin{equation*}
d^\circ(F\circ H)= \begin{cases}
3 &\text{ if }  p=2 \\
3 \text{ or } 4 &\text{ if }  p>2.
\end{cases}
\end{equation*}
To validate the above result, we first need to find a polynomial expression of $F\circ H$ when $d^\circ(F)=2$.
 Let
 $$
  F(x)=F_{\text{\tiny DO}}(x)+ F_{\ell}(x)+F(0),
 $$
 where $F_\ell$ is the linear part of $F$ and
 $$
 F_{\text{\tiny DO}}(x)=\sum_{i\leq j} a_{i,j}x^{p^i+p^j}
    \qquad (\text{When } p=2, \text{ the ordering is strict, i.e., } i<j)
 $$
 is the DO-part of $F(x)$.

\begin{prop}\label{key-prop2}
Let $F(x)$ be any quadratic $(n,n)$-function and let $H(x)=x+\gamma\Tr(\alpha x+\beta F(x))$
 with $\alpha, \beta, \gamma \in \Fpn$. Then, letting $\eps=\Tr(\alpha x+\beta F(x))$, one has
 \begin{align*}
  F\circ H(x) =F(x)+\eps D_\gamma F(x)
              + F_{\text{\tiny DO}}(\gamma)(\eps^2-\eps),
  \end{align*}
  where $F_{\text{\tiny DO}}$ is the DO-part of $F$. More precisely,
   \begin{enumerate}
   \item[1.] If $p=2$, then
            \begin{align}
            F\circ H(x) & =F(x)+\eps D_\gamma F(x). \label{eqalg3}
            \end{align}
   \item[2.] If $p>2$, then
               \begin{align}
            F\circ H(x) & =F(x)+\eps D_\gamma F(x) +F_{\text{\tiny DO}}(\gamma)(\eps^2-\eps). \label{eqalg4}
             \end{align}
  \end{enumerate}
\end{prop}

\noindent
\begin{proof}
 Using $H(x)=x+\gamma\eps$,
 \begin{align}
 F\circ H(x)&=F(x+\gamma\eps)=F_{\text{\tiny DO}}(x+\gamma\eps)+ F_{\ell}(x+\gamma\eps)+F(0). \label{fh1}
 \end{align}
Since
\begin{align}
 F_{\text{\tiny DO}}(x+\gamma\eps)&=\sum_{i\leq j} a_{i,j}(x+\gamma\eps)^{p^i+p^j}
    =\sum_{i\leq j} a_{i,j}(x^{p^i}+\gamma^{p^i}\eps)(x^{p^j}+\gamma^{p^j}\eps) \notag \\
    &=\sum_{i\leq j} a_{i,j}\left(  x^{p^i+p^j}+
    (\gamma^{p^j}x^{p^i}+\gamma^{p^i}x^{p^j}+\gamma^{p^i+p^j})\eps+\gamma^{p^i+p^j}(\eps^2-\eps) \right)\notag \\
    &=F_{\text{\tiny DO}}(x)+ \eps D_\gamma F_{\text{\tiny DO}}(x)+ F_{\text{\tiny DO}}(\gamma)(\eps^2-\eps), \label{fh2}
 \end{align}
 Combining Equation \eqref{fh1} and \eqref{fh2}, one has
  \begin{align}
   F\circ H(x)&=F_{\text{\tiny DO}}(x+\gamma\eps)+ F_{\ell}(x+\gamma\eps)+F(0) \notag \\
            &=F_{\text{\tiny DO}}(x)+\eps D_\gamma F_{\text{\tiny DO}}(x)+ F_{\text{\tiny DO}}(\gamma)(\eps^2-\eps)
            + F_{\ell}(x)+F_\ell(\gamma)\eps+F(0) \notag \\
            &=\{F_{\text{\tiny DO}}(x)+ F_{\ell}(x)+F(0)\}+\eps \{D_\gamma F_{\text{\tiny DO}}(x)+F_\ell(\gamma)\}
              + F_{\text{\tiny DO}}(\gamma)(\eps^2-\eps) \notag \\
            &=F(x)+\eps D_\gamma F(x)
              + F_{\text{\tiny DO}}(\gamma)(\eps^2-\eps). \notag
  \end{align}
  \end{proof}

\medskip
 \noindent
 From the above Proposition, to show $d^\circ (F\circ H)>2$, it is enough to show that the algebraic degrees of
 $\eps D_\gamma F(x)$ and $\eps^2$ are greater than two. In the following section, we will show that
 \begin{align*}
  d^\circ(\eps D_\gamma F(x))&=3 \quad \text{ for any } p, \quad \text{ and } \\
   \quad d^\circ(\eps^2) &=4 \quad \text{ for odd } p,
 \end{align*}
 are satisfied under certain mild restrictions on $F$.

\section{Algebraic Degrees of Products of Quadratic Functions}\label{sec-degree}
For quadratic $F(x)\in \Fpn[x]$, recall that $L_\gamma(x)=F(x+\gamma)-F(x)-(F(\gamma)-F(0))$ is bilinear with respect
to $\gamma, x \in \Fpn.$ That is, $L_{\gamma_1+\gamma_2}(x)=L_{\gamma_1}(x)+L_{\gamma_2}(x)$ and
$L_{\gamma}(x+y)=L_{\gamma}(x)+L_{\gamma}(y)$. Moreover $L_\gamma(x)$ is symmetric in the sense that
$L_\gamma(\beta)=L_\beta(\gamma)$.

\begin{lemma}\label{lemF3}
Let $F(x)\in \Fpn[x]$ be a quadratic $(n,n)$-function and let
$$F_3(x)=\Tr(\beta F(x))D_\gamma F(x)$$
 with $\beta, \gamma \in
\Fpn$. Then for any $\gamma_1, \gamma_2, \gamma_3 \in \F_{p^n}$, we have
$$
D_{\gamma_3} D_{\gamma_2} D_{\gamma_1} F_3(x)=\Tr(\beta L_{\gamma_2}(\gamma_3))L_\gamma(\gamma_1) +\Tr(\beta
L_{\gamma_1}(\gamma_3))L_\gamma(\gamma_2) +\Tr(\beta L_{\gamma_1}(\gamma_2))L_\gamma(\gamma_3).
$$
\end{lemma}

\noindent
\begin{proof}
From $F_3(x)=\Tr(\beta F(x))D_\gamma F(x)= \Tr(\beta F(x))\left(L_\gamma(x)+(F(\gamma)-F(0))\right)$, one gets
$D_{\gamma_3} D_{\gamma_2} D_{\gamma_1} F_3(x)=D_{\gamma_3} D_{\gamma_2} D_{\gamma_1} \left(\Tr(\beta F(x))
L_\gamma(x)\right)$ because $d^\circ(\Tr(\beta F(x))\leq 2$. Now
\begin{align*}
D_{\gamma_1} \left(\Tr(\beta F(x)) L_\gamma(x)\right)
 &= \Tr(\beta F(x+\gamma_1))L_\gamma(x+\gamma_1)-\Tr(\beta F(x))L_\gamma(x) \\
  &= \Tr\left(\beta F(x+\gamma_1)-\beta F(x)\right)L_\gamma(x)+\Tr(\beta F(x+\gamma_1))L_\gamma(\gamma_1) \\
  &= \Tr\left(\beta L_{\gamma_1}(x)+\beta (F(\gamma_1)-F(0))\right)L_\gamma(x)+\Tr(\beta F(x+\gamma_1))L_\gamma(\gamma_1) \\
  &= \Tr\left(\beta L_{\gamma_1}(x)\right)L_\gamma(x)+\Tr(\beta F(x+\gamma_1))L_\gamma(\gamma_1) +\Tr(\beta
  (F(\gamma_1)-F(0)))L_\gamma(x).
\end{align*}
Since $d^\circ(L_\gamma(x))\leq 1$, one has
$$
D_{\gamma_3}D_{\gamma_2}D_{\gamma_1} \left(\Tr(\beta F(x)) L_\gamma(x)\right)=
  D_{\gamma_3}D_{\gamma_2} \left[ \Tr\left(\beta L_{\gamma_1}(x)\right)L_\gamma(x)+\Tr(\beta F(x+\gamma_1))L_\gamma(\gamma_1)
  \right],
$$
where
\begin{align*}
D_{\gamma_2} & \left[ \Tr\left(\beta L_{\gamma_1}(x)\right)L_\gamma(x)+\Tr(\beta F(x+\gamma_1))L_\gamma(\gamma_1)
  \right] \\
  =\,  & \Tr\left(\beta L_{\gamma_1}(x+\gamma_2)\right)L_\gamma(x+\gamma_2)+\Tr(\beta
  F(x+\gamma_1+\gamma_2))L_\gamma(\gamma_1)\\
  &- \Tr\left(\beta L_{\gamma_1}(x)\right)L_\gamma(x)-\Tr(\beta F(x+\gamma_1))L_\gamma(\gamma_1) \\
  =\, & \Tr\left(\beta L_{\gamma_1}(x)\right)L_\gamma(\gamma_2)
   +\Tr\left(\beta L_{\gamma_1}(\gamma_2)\right)L_\gamma(x) +\Tr\left(\beta L_{\gamma_1}(\gamma_2)\right)L_\gamma(\gamma_2)\\
    &+\Tr\left(\beta L_{\gamma_2}(x+\gamma_1)+\beta(F(\gamma_2)-F(0))\right)L_\gamma(\gamma_1).
\end{align*}
Since $D_{\gamma_3}$ kills all the constant terms, we get
$$
D_{\gamma_3}D_{\gamma_2}D_{\gamma_1} \left(\Tr(\beta F(x)) L_\gamma(x)\right)=
 \Tr(\beta L_{\gamma_1}(\gamma_3))L_\gamma(\gamma_2) +\Tr(\beta L_{\gamma_1}(\gamma_2))L_\gamma(\gamma_3)
 +\Tr(\beta L_{\gamma_2}(\gamma_3))L_\gamma(\gamma_1).
$$
\end{proof}

\begin{prop}\label{cubic-degree}
 Suppose that $F$ is a quadratic $(n,n)$-function with $n\geq 4$
  and suppose that
  $$
  | \cW_{F}|\overset{\text{\rm def}}{=} \max_{b \in \Fpn^\times} | \cW_{\Tr(b F)}| < p^n \quad \textrm{and} \quad \Delta_F\leq p^{n-3}.
  $$
 Let $\eps=\Tr(\alpha x +\beta F(x))$ with $\alpha\in \Fpn$ and $\beta\in \Fpn^\times$.
  Then, one has $d^\circ (\eps D_\gamma F)=3$ for any $\gamma\in \Fpn^\times$.
 % Let $H(x)=x+\gamma\Tr(\alpha x+\beta F(x))$ and $F\circ H(x)=F(x)+\Tr(\alpha x+\beta F(x))D_\gamma
 %F(x)$. Then $d^\circ (F\circ H)=3$.
\end{prop}

\noindent
 \begin{proof}
Since $d^\circ(F)=2$ and $d^\circ(\Tr(\alpha x)D_\gamma F(x))\leq 2$, to prove the proposition, it is enough to show
that $d^\circ(F_3)=3$ where $F_3(x)=\Tr(\beta F(x))D_\gamma F(x)$. On the contrary, suppose that $d^\circ(F_3)\leq 2$.
Then for any $\gamma_1, \gamma_2, \gamma_3 \in \Fpn$, one has $D_{\gamma_3} D_{\gamma_2} D_{\gamma_1} F_3(x)=0$ for all
$x\in \Fpn$. On the other hand,  by Lemma \ref{lemF3}, one has
 \begin{align}\label{3rd_derivative}
  D_{\gamma_3} D_{\gamma_2} D_{\gamma_1}
F_3(x)=\Tr(\beta L_{\gamma_2}(\gamma_3))L_\gamma(\gamma_1) +\Tr(\beta L_{\gamma_1}(\gamma_3))L_\gamma(\gamma_2)
+\Tr(\beta L_{\gamma_1}(\gamma_2))L_\gamma(\gamma_3),
 \end{align}
which is a constant independent of $x$ (but depending on $\gamma_1,\gamma_2,\gamma_3$). We will show that the
derivative is nonzero for some suitable choices of $\gamma_1,\gamma_2$ and $\gamma_3$. From the assumption, one gets
\begin{align}
p^s\overset{\textrm{def}}{=}|\Ker L_\gamma|\leq \Delta_F\leq p^{n-3} \label{n-s}.
\end{align}
Also, one has the following isomorphism between vector spaces over $\Fp$,
$$
\Fpn /\Ker L_\gamma \cong \Ima L_\gamma,
$$
which implies that $|\Ima L_\gamma|=p^{n-s}\geq p^3$ by the Equation (\ref{n-s}). Therefore $\Ima L_\gamma$ is a vector
subspace of dimension $n-s\geq 3$ over $\Fp$ and there exist at least three independent (nonzero) vectors in $\Ima
L_\gamma$. Moreover, we have a disjoint (additive) coset decomposition of $\Fpn$ as
\begin{align*}
\Fpn = \cup_{i=0}^{p^{n-s}-1}=K_0 \cup K_1 \cup \cdots \cup K_{p^{n-s}-1} \qquad (K_i=\theta_i +\Ker L_\gamma)
\end{align*}
with $\theta_0=0$ and $K_0=\Ker L_\gamma$, and different cosets maps to different
 elements under $L_\gamma$, i.e., $L_\gamma(\theta_i+x)=L_\gamma (\theta_i)$ for all $x\in \Ker L_\gamma$.
Recall that the linear kernel of $\Tr(\beta F)$
 \begin{align}
 \mathcal{LK}&=\mathcal{LK}_{\Tr(\beta F)} \notag \\
 &=\{\gamma\in \F_{p^n} :
\Tr(\beta F(x+\gamma)-bF(x))=\Tr(\beta F(\gamma)-bF(0)) \textrm{ for all } x\in \mathbb
F_{p^n}  \} \notag \\
 &=\{\gamma\in \F_{p^n} :
\Tr(\beta L_\gamma(x))=0 \textrm{ for all } x\in \mathbb F_{p^n}  \}
 =\{\gamma \in \Fpn : \Ima L_\gamma \subset \langle \beta \rangle^{\perp} \} \label{LK}
\end{align}
is a subspace of dimension (say) $k_\beta$ such that $| \cW_{\Tr(\beta F)}|= p^\frac{n+k_\beta}{2}$. Our assumption of
$$ | \cW_{F}|= \max_{b \in \Fpn^\times} | \cW_{\Tr(b F)}| < p^n $$
says that the dimension $k_b$ of linear kernel $\mathcal{LK}_{\Tr(b F)}$ strictly smaller than $n$, i.e., $k_b < n$ for
all $b \in \Fpn^\times$. (When $p=2$, an equivalent assumption is $\NL(F)\neq 0$ (or $\LL(F)<2^n$).) Therefore we get
$$
\dim \mathcal{LK}=\dim \mathcal{LK}_{\Tr{(\beta F)}} =k_\beta \leq n-1,
$$
which implies
$$
 |\mathcal{LK}|\leq p^{n-1}  \quad \textrm{i.e.,}\,\, |\Fpn \setminus \mathcal{LK}| \geq p^n-p^{n-1}.
$$
We claim that the number of $1\leq i\leq p^{n-s}-1$ such that $K_i \subset \LK$ (i.e., the number of nonzero cosets
$K_i$ contained in $\LK$) is at most $p^{n-s-1}$. If not, there are $p^{n-s-1}+1$ different $K_i$ such that $K_i\subset
\LK$. Since all cosets are disjoint, one gets $p^{n-1}+p^s=(p^{n-s-1}+1)p^s\leq |\LK|\leq p^{n-1}$ which is a
contradiction. Therefore, the number of $K_i (1\leq i\leq p^{n-s}-1)$ such that $K_i\not\subset \LK$ is at least
$p^{n-s}-1-p^{n-s-1}=(p-1)p^{n-s-1}-1$. Pick any nonzero coset $K_i \not\subset \LK$ and choose any $\gamma_1 \in
K_i\setminus \LK$. Since there are $p-1$ cosets $tK_i=t\theta_i+\Ker L_\gamma$ $(1\leq t\leq p-1)$ such that $tK_i$ and
$K_i$ are dependent over $\Fp$ and since $p-1 < (p-1)p^2-1 \leq (p-1)p^{n-s-1}-1$, there exists $K_j (\neq tK_i)$ such
that $K_j\not\subset \LK$. Now, choose any $\gamma_2 \in K_j\setminus \LK$. Then one has $\gamma_1, \gamma_2 \not\in
\LK$ and, via the isomorphism $\Fpn/\Ker L_\gamma \cong \Ima L_\gamma$, one gets two linearly independent vectors
$L_\gamma(\gamma_1), L_\gamma(\gamma_2)$  over $\Fp$. Since $\gamma_1, \gamma_2 \not\in \LK=\LK_{\Tr(\beta F)}$, in
view of the definition of $\LK$ in (\ref{LK}), there exist $x_1, x_2 \in \Fpn^\times$ such that
\begin{align}\label{eqbetaL}
 \Tr(\beta L_{\gamma_1}(x_1))=1  \quad \textrm{and}\quad  \Tr(\beta L_{\gamma_2}(x_2))=1.
\end{align}
So far, $\gamma_1$ and $\gamma_2$ are determined and we will choose $\gamma_3$ to make $D_{\gamma_3} D_{\gamma_2}
D_{\gamma_1} F_3(x)$ nonvanishing. From Equation (\ref{3rd_derivative}), we see that $D_{\gamma_3} D_{\gamma_2}
D_{\gamma_1} F_3(x)$ is a $\Fp$-linear combination of three vectors $L_\gamma(\gamma_j)\,\, (1\leq j\leq 3)$, where
$L_\gamma(\gamma_1)$ and $L_\gamma(\gamma_2)$ are linearly independent over $\Fp$.
 If $\Tr(\beta L_{\gamma_1}(\gamma_2))=0$, then from Equation (\ref{eqbetaL}), choose $\gamma_3$ satisfying
 $\Tr(\beta L_{\gamma_2}(\gamma_3))=1$. Then one has
 $$
 D_{\gamma_3} D_{\gamma_2} D_{\gamma_1}
F_3(x)=L_\gamma(\gamma_1) +\Tr(\beta L_{\gamma_1}(\gamma_3))L_\gamma(\gamma_2)
$$
which is nonzero. If $\Tr(\beta L_{\gamma_1}(\gamma_2))\neq 0$, since $\dim \Ima L_\gamma =n-s\geq 3$, one can choose
$\gamma_3\in \Fpn$ such that $L_\gamma(\gamma_3)$ is not spanned by $\{L_\gamma(\gamma_1), L_\gamma(\gamma_2)\}$. Then
$$
D_{\gamma_3} D_{\gamma_2} D_{\gamma_1} F_3(x)=\Tr(\beta L_{\gamma_2}(\gamma_3))L_\gamma(\gamma_1) +\Tr(\beta
L_{\gamma_1}(\gamma_3))L_\gamma(\gamma_2) +\Tr(\beta L_{\gamma_1}(\gamma_2))L_\gamma(\gamma_3)
$$
is a nonzero vector of $\Fpn$.
 \end{proof}

\begin{remark}
Please be aware that if any of the two conditions $ | \cW_{F}|<p^n$ and $\Delta_F\leq p^{n-3}$ is not satisfied, then
there are some functions $F$ such that $d^\circ (\eps D_\gamma F)<3$. Let us give such examples for a binary case.
 Let $F(x)=x\Tr(x)$ be a quadratic $(n,n)$-function on $\Fbn$. It is clear that $\NL(F)=0$ since $\Tr(F(x))=\Tr(x)$.
  Let $\gamma=1=\beta$ and $\alpha=0$. Then, since $d^\circ(\Tr(F))=1$, one gets $d^\circ \left(\Tr(F(x)) D_1 F(x)\right)<3$.
   Also, for the same $F(x)=x\Tr(x)$, by choosing $\gamma \in \Fbn^\times$ and $\beta\in \Fbn\setminus\F_2$ satisfying
   $\Tr(\gamma)=0$ and $\Tr(\beta\gamma)=0$, one gets $\Tr(\beta D_\gamma F(x))=0$ for all $x\in\Fbn$ where
   $\Tr(\beta F(x))=\Tr(x)\Tr(\beta x)$ and $D_\gamma F(x)=\gamma \Tr(x)$. Therefore one has
    $d^\circ \left(\Tr(\beta F) D_\gamma F\right)=2$ even if $d^\circ (\Tr(\beta F))=2$ and $d^\circ (D_\gamma F)=1$.
   For another example, let $F(x)=x^5=x^{2^2+1} \in \F_{2^4} [x]$ be a Gold function on $\F_{2^4}$. Note that
   $\Delta_F=4>2=2^{n-3}$. Then one has $\Tr(F(x))=x^5+x^{10}+x^{5}+x^{10}=0$ for all $x$, which implies that
  $d^\circ \left(\Tr(F(x)) D_\gamma F(x)\right)=0$ for any $\gamma$. For another extreme case, let $F(x)=x^{2^i+1}+x^{2^j+1}+x^{2^i+2^j}$
  (with $i\neq j$) be a quadratic $(n,n)$-function on $\Fbn$. It is clear that $\Delta_F=2^n$ since
   $D_1F(x)=x^{2^i}+x+1+x^{2^j}+x+1+x^{2^i}+x^{2^j}+1=1$. Therefore one has
   $d^\circ \left(\Tr(\beta F) D_1 F\right)<3$ for any $\beta$.
\end{remark}

\begin{lemma}\label{quartic}
    Let $p$ be an odd prime and let $F(x)\in \Fpn[x]$ be a quadratic $(n,n)$-function.
    Let $F_4(x)=\Tr(\beta F(x))^2$. Then in a similar manner as in the proof of Lemma \ref{lemF3},
    for any $\gamma, \gamma_1, \gamma_2, \gamma_3 \in \Fpn$,  one has
    \begin{align*}
 \textstyle\frac{1}{2}D_{\gamma_3}D_{\gamma_2}&D_{\gamma_1}D_{\gamma} F_4(x) \\
   &= \Tr(\beta L_{\gamma_2}(\gamma_3))\Tr(\beta L_\gamma(\gamma_1)) +
 \Tr(\beta L_{\gamma_1}(\gamma_3))\Tr(\beta L_\gamma(\gamma_2))
 +\Tr(\beta L_{\gamma_1}(\gamma_2))\Tr(\beta L_\gamma(\gamma_3)).
    \end{align*}
\end{lemma}
\begin{proof}
One has
 \begin{align*}
 D_\gamma F_4(x)&=\Tr(\beta F(x+\gamma))^2-\Tr(\beta F(x))^2   \\
  &=\{\Tr(\beta F(x+\gamma))-\Tr(\beta  F(x))+2\Tr(\beta F(x))\}\{\Tr(\beta F(x+\gamma))-\Tr(\beta F(x))\} \\
  &=\{\Tr(\beta D_\gamma F(x))+ 2\Tr(\beta F(x))\}\Tr(\beta D_\gamma F(x)).
 \end{align*}
Since $d^\circ \{ \Tr(\beta D_\gamma F)^2\}\leq 2$ and since $p\neq 2$, one has
\begin{align}
 \textstyle\frac{1}{2}D_{\gamma_3}D_{\gamma_2}D_{\gamma_1}D_{\gamma} F_4(x)
  = D_{\gamma_3}D_{\gamma_2}D_{\gamma_1} \{\Tr(\beta F(x))\Tr(\beta D_\gamma F(x))\}. \label{qt1}
\end{align}
From Lemma \ref{lemF3},
\begin{align*}
D_{\gamma_3} & D_{\gamma_2}  D_{\gamma_1} \{ \Tr(\beta F(x))D_\gamma F(x) \} \\
     &=  \Tr(\beta L_{\gamma_2}(\gamma_3))L_\gamma(\gamma_1)
+\Tr(\beta L_{\gamma_1}(\gamma_3))L_\gamma(\gamma_2) +\Tr(\beta
L_{\gamma_1}(\gamma_2))L_\gamma(\gamma_3).
\end{align*}
Multiply $\beta$ to both sides of the above equation and take the trace, then since the trace map commutes with
differentials,
\begin{align}
D_{\gamma_3} & D_{\gamma_2}  D_{\gamma_1} \{ \Tr(\beta F(x))\Tr(\beta D_\gamma F(x)) \} \label{qt2} \\
 &=  \Tr(\beta L_{\gamma_2}(\gamma_3))\Tr(\beta L_\gamma(\gamma_1))
+\Tr(\beta L_{\gamma_1}(\gamma_3))\Tr(\beta L_\gamma(\gamma_2))
 +\Tr(\beta L_{\gamma_1}(\gamma_2))\Tr(\beta L_\gamma(\gamma_3)). \notag
\end{align}
Combining Equations \eqref{qt1} and \eqref{qt2}, we get the desired result.
\end{proof}

\begin{prop}\label{quartic-degree}
Let $p$ be an odd prime and let $F$ be a quadratic $(n,n)$-function with $n\geq 4$.
    Suppose that
    $$|\cW_F|\overset{\text{\rm def}}{=}
    \max_{b\in \Fpn^\times}|\cW_{\Tr(b F)}| \leq p^{n-1}.$$
  Let $\eps=\Tr(\alpha x +\beta F(x))$ with $\alpha\in \Fpn$ and $\beta\in \Fpn^\times$.
  Then, one has $d^\circ (\eps^2)=4$.
\end{prop}

\noindent
\begin{proof}
In a similar manner as in Proposition \ref{cubic-degree}, it is enough to show $d^\circ(F_4)=4$ where $F_4(x)=\Tr(\beta
F(x))^2$. Let $\tau=\frac{1}{2}D_{\gamma_3}D_{\gamma_2}D_{\gamma_1}D_{\gamma} F_4(x)$. If $d^\circ (F_4)\leq 3$, then
one has $\tau=0$ for all $\gamma, \gamma_1, \gamma_2$ and $\gamma_3$ in $\Fpn$. However, we will show that $\tau\neq 0$
for some suitably chosen $\gamma, \gamma_1, \gamma_2, \gamma_3$. Letting $\gamma_1=\gamma$ and $\gamma_3=\gamma_2$,
\begin{align}
 \tau &=\Tr(\beta L_{\gamma_2}(\gamma_3))\Tr(\beta L_\gamma(\gamma_1)) +
 \Tr(\beta L_{\gamma_1}(\gamma_3))\Tr(\beta L_\gamma(\gamma_2))
  +\Tr(\beta L_{\gamma_1}(\gamma_2))\Tr(\beta L_\gamma(\gamma_3)) \notag \\
 &=\Tr(\beta L_{\gamma_2}(\gamma_2))\Tr(\beta L_\gamma(\gamma))+2\Tr(\beta L_{\gamma}(\gamma_2))^2. \label{quartic-eq}
\end{align}
 Letting $\dim \LK_{\Tr(\beta F)}=k_\beta$
 ($\dim \LK_{\Tr(\beta F)}$ is the dimension of the linear kernel of $\Tr(\beta F)$, one has
 $p^{n+{k_\beta}}\leq |\cW_F|^2\leq p^{2n-2}$, where the first $\leq$ comes from Proposition \ref{pbentwalsh}
  and the second $\leq$ comes from the assumption. Therefore $k_\beta\leq n-2$, or equivalently,
 $|\LK_{\Tr(\beta F)}|=p^{k_\beta}\leq p^{n-2}$.
Now, choosing any $\gamma \in \Fpn \setminus \LK_{\Tr(\beta F)}$,  $\Tr(\beta L_\gamma(x))$
 is a balanced (linear) function on $\Fpn$, i.e., $\Tr(\beta L_\gamma(x))$ takes $j\in \Fp$ exactly $p^{n-1}$ times
 for every $0\leq j<p$. Also, for all $\gamma'\in \LK_{\Tr(\beta F)}$, one has
 $\Tr(\beta L_\gamma(\gamma'))=\Tr(\beta L_{\gamma'}(\gamma))=0$.
 Since all elements  $\gamma'$ of $\LK_{\Tr(\beta F)}$ satisfies $\Tr(\beta L_\gamma(\gamma'))=0$
  with $|\LK_{\Tr(\beta F)}|\leq p^{n-2}$, and since
 $\Tr(\beta L_\gamma(x))$ is a balanced function $(\because \gamma\notin \LK_{\Tr(\beta F)})$,
 $\Tr(\beta L_\gamma(x))=0$ is satisfied for {\it at least} $p^{n-2}$ values of $x\in \Fpn \setminus \LK_{\Tr(\beta
 F)}$,
 and $\Tr(\beta L_\gamma(x))=j$ is satisfied for {\it exactly} $p^{n-1}$ values of $x\in \Fpn \setminus \LK_{\Tr(\beta F)}$
 for every $1\leq j\leq p-1$. Now, from Equation \eqref{quartic-eq}, we consider two cases:
  If $\Tr(\beta L_\gamma(\gamma))=0$ for some $\gamma \in \Fpn \setminus \LK_{\Tr(\beta F)}$, then choose
 $\gamma_2\in \Fpn \setminus \LK_{\Tr(\beta F)}$ satisfying $\Tr(\beta L_\gamma(\gamma_2))=1$
  (there are exactly $p^{n-1}$ such $\gamma_2$ in $\Fpn \setminus \LK_{\Tr(\beta F)}$).
   Then one has $\tau=2\Tr(\beta L_{\gamma}(\gamma_2))^2=2\neq 0$.
   If $\Tr(\beta L_\gamma(\gamma))\neq 0$ for all $\gamma \in \Fpn \setminus \LK_{\Tr(\beta F)}$,
   then choose $\gamma_2\in \Fpn \setminus \LK_{\Tr(\beta F)}$ satisfying $\Tr(\beta L_\gamma(\gamma_2))=0$
  (there are at least $p^{n-2}$ such $\gamma_2$ in $\Fpn \setminus \LK_{\Tr(\beta F)}$).
   Then one has $\tau=\Tr(\beta L_{\gamma_2}(\gamma_2))\Tr(\beta L_\gamma(\gamma))\neq 0$.
\end{proof}

\medskip
\noindent It also turns out that our constructed function $F\circ H$ is EA-inequivalent to any power function.
 To prove that, we recall Proposition $4$ in \cite{BCP06}.
\begin{prop}\label{buda-prop}
 Let $F$ be an $(n,n)$-function on $\Fpn$. If there exists $c\in \Fpn^\times$ such that
  $d^\circ(\Tr(cF)) \not\in \{0,1,d^\circ(F)\}$, then $F$ is EA-inequivalent to any power function.
\end{prop}
\begin{proof}
 The case for $p=2$ is proven in \cite{BCP06}. The case for arbitrary characteristic is exactly the same
  as in \cite{BCP06}.
\end{proof}

\begin{theorem}{\bf(Main Theorem)}\label{mainthm}
 Let $F_0(x)$ be a quadratic $(n,n)$-function on $\Fpn$ with $n\geq 4$. Then,  one concludes the
 followings.
 \begin{enumerate}
 \item[1.] When $p=2$: Suppose that $|\cW_{F_0}|<2^n$ (i.e., $\NL(F_0)\neq 0$) and $\Delta_{F_0}\leq 2^{n-3}$.
       Then for any $\gamma\in\Fbn^\times$, one can choose $\beta\in\Fbn^\times$ such that $\Tr(\beta D_\gamma F(x))=0$
           for all $x\in \Fbn$, where $F(x)=F_0(x)+cx$ and $c\in \Fbn$ is determined by
           the condition $\Tr(c\beta\gamma)=\Tr(\beta (F_0(\gamma)+F_0(0)))$.
            Choosing $\alpha\in \Fbn$ satisfying $\Tr(\alpha\gamma)=0$, the function
            $$
              F\circ H(x)=F(x)+\eps D_\gamma F(x) \quad (\,\,\text{with }\eps=\Tr(\alpha x +\beta F(x)) \,\,)
            $$
         has an algebraic degree $3$ and satisfies $F\circ H \eqccz F_0$ and $F\circ H \noteqea F_0$. Moreover
         $F\circ H$ is EA-inequivalent to any power function.

 \item[2.] When $p>2$: Suppose that $|\cW_{F_0}|\leq p^{n-1}$ and $1<\Delta_{F_0}\leq p^{n-3}$.
         Then there exist $\gamma, \beta \in\Fpn^\times$ such that $\Tr(\beta D_\gamma F(x))=0$
           for all $x\in \Fpn$, where $F(x)=F_0(x)+cx$ and $c\in \Fpn$ is determined by
           the condition $\Tr(c\beta\gamma)=-\Tr(\beta (F_0(\gamma)-F_0(0)))$.
            Choosing $\alpha\in \Fpn$ satisfying $\Tr(\alpha\gamma)=-2$, the function
            $$
              F\circ H(x)=F(x)+\eps D_\gamma F(x) +F_{\text{\tiny DO}}(\gamma)(\eps^2-\eps)
                \quad (\,\,\text{with }\eps=\Tr(\alpha x +\beta F(x)) \,\,)
            $$
         satisfies $F\circ H \eqccz F_0$ and $F\circ H \noteqea F_0$, where
         $d^\circ (F\circ H)=3$ if $F_{\text{\tiny DO}}(\gamma)=0$, and
         $d^\circ (F\circ H)=4$ if $F_{\text{\tiny DO}}(\gamma)\neq 0$.
         Furthermore, when $F_{\text{\tiny DO}}(\gamma)=0$, then
         $F\circ H$ is EA-inequivalent to any power function. Also, when $F_{\text{\tiny DO}}(\gamma)\neq 0$,
         by assuming that $F_0$ is a DO-polynomial (one can always assume this up to EA-equivalence),
        $F\circ H$ is EA-inequivalent to any power function if $\Tr(c\beta\gamma)=0$.
 \end{enumerate}

\end{theorem}

\noindent
\begin{proof}
 Existence of such $\gamma, \beta \in \Fpn^\times$ is guaranteed by Proposition \ref{LS-quad}. Existence of
 $\alpha, c \in \Fpn$ is guaranteed by Proposition \ref{key-prop}. CCZ-equivalence of constructed
  $F\circ H$ is provided by Theorem \ref{thm1}. Explicit expressions of $F\circ H$ for each characteristic are provided
  by Proposition \ref{key-prop2}. EA-inequivalence (or $d^\circ (F\circ H) > 2$) is guaranteed by Propositions
  \ref{cubic-degree} and \ref{quartic-degree}. Therefore
   one only needs to prove the last part of each assertion.

\medskip

 \noindent
 ({\it 1.}) Since $\NL(F)\neq 0$ and $\beta\neq 0$, one has $d^\circ(\Tr(\beta F))=2$. Also from
     $$
      \Tr(\beta [F\circ H(x)])=\Tr(\beta F(x))+\eps \Tr(\beta D_\gamma F(x))=\Tr(\beta F(x)),
     $$
     one gets $d^\circ(\Tr(\beta [F\circ H]) )=d^\circ( \Tr(\beta F))=2$. Therefore using Proposition \ref{buda-prop},
     one concludes that $F\circ H \noteqea x^s$ for any $s\in \Z$.

 \smallskip
  \noindent
  ({\it 2.}) The proof of the case $F_{\text{\tiny DO}}(\gamma)=0$ is exactly same to the assertion {\it 1.}
   If $F_{\text{\tiny DO}}(\gamma)\neq 0$, then since $F_0$ is a DO-polynomial,
    one has $F_0(x)=F_{\text{\tiny DO}}(x)$. Therefore
   using
    $$0=\Tr(c\beta\gamma)=-\Tr(\beta (F_0(\gamma)-F_0(0)))
     =-\Tr(\beta F_0(\gamma))=-\Tr(\beta F_{\text{\tiny DO}}(\gamma)),
     $$
     one gets
   \begin{align*}
   \Tr(\beta[F\circ H(x)])
    =\Tr(\beta F(x))+\eps \Tr(\beta D_\gamma F(x)) +(\eps^2-\eps) \Tr(\beta F_{\text{\tiny DO}}(\gamma))
     =\Tr(\beta F(x)).
   \end{align*}
   And the rest is similar to the assertion {\it 1.}
\end{proof}

\begin{corollary}
 Let $n$ be even and let $F(x)$ be a quadratic APN function from $\Fbn$ to itself. Then one has
 $\clsEA(F)\subsetneqq \clsCCZ(F)$, i.e., there exists $F' \in \clsCCZ(F)\setminus \clsEA(F)$.
\end{corollary}
\begin{proof}
Since there is no quadratic APN permutation for even $n$ (See \cite{SZZ94}.), $F'=F\circ H$ is non-obvious function
which is CCZ-equivalent but EA-inequivalent to $F$.
\end{proof}

%\begin{remark}
% When $n$ is even, to the authors' knowledge, {\color{red} (???)} $\clsEA(F_0)\subsetneqq \clsCCZ(F_0)$ In the case of
%\end{remark}

%\begin{remark}
%{\color{red} My question to Dr. Koo and Dr. Jeong :} In the above proof where $p>2$, how many of such $\gamma$
% (such that $D_\gamma F$ is not one to one) exist ? Please kindly use SAGE with $F=$ non-Gold function. Because, in
%5 Gold case, if $n=$even, $D_\gamma F$ is not one to one for every $\gamma\in \Fpn^\times$.
%\end{remark}

\section{Concrete Examples}\label{sec-concrete}

\subsection{Gold Functions for $p=2$}\label{subsec-goldbinary}
For a binary case,  Budaghyan et al. \cite{BCP06} constructed examples of $(n,n)$-function $F'$ such that
$$
 G \noteqea F' \quad \text{ but }  \quad G\eqccz F'
$$
where $G(x)=x^{2^i+1}$ (with $\gcd(i,n)=1$ and $n\geq 4$) is the Gold function. Such $G$ is an AB(almost bent) function
for odd $n$, and is an APN function for even $n$. We summarize, in Table \ref{table-buda}, $F'(x)$ and the
corresponding linear permutation $\caliL$ on $\Fbn\times \Fbn$ (for CCZ-equivalence) used in \cite{BCP06}.

\begin{table}[ht]
\caption{$F'(x)$ and $\caliL$ used in \cite{BCP06}}
\smallskip
\centering
 \renewcommand{\arraystretch}{1.1}
\begin{tabular}{c | c | c | c}  % centered columns (4 columns)
\hline\hline
$n$ & $F'(x)$ & $\caliL (x,y)$ & $d^\circ (F')$ \\ [0.5ex] % inserts table
%heading
\hline
odd & $x^{2^i+1}+(x^{2^i}+x)\Tr(x^{2^i+1}+x)$ & $(x+\Tr(x+y),y+\Tr(x+y))$ & $3$  \\
even & $x^{2^i+1}+(x^{2^i}+x+1)\Tr(x^{2^i+1})$ & $(x+\Tr(y),y)$ & $3$ \\
\hline\hline
\end{tabular}
\label{table-buda}  % is used to refer this table in the text
\end{table}

\begin{table}[ht]
\caption{Our Rediscovery}
\smallskip
\centering
 \renewcommand{\arraystretch}{1.1}
\begin{tabular}{c | c | c | c | c |c  }  % centered columns (4 columns)
\hline\hline
$n$ & $F\circ H$ & $\alpha$& $\beta$ & $\gamma$ &  $c$ \\ [0.5ex] % inserts table
%heading
\hline
odd & $x^{2^i+1}+x+(x^{2^i}+x)\Tr(x^{2^i+1}+x)$ & $0$ & $1$ & $1$ & $1$  \\
\hline
even & $x^{2^i+1}+(x^{2^i}+x+1)\Tr(x^{2^i+1})$ & $0$ & $1$ & $1$ & $0$ \\
even & $x^{2^i+1}+x+(x^{2^i}+x)\Tr(x^{2^i+1}+x)$ & $0$ & $1$ & $1$ & $1$ \\
\hline\hline
\end{tabular}
\label{table-buda}  % is used to refer this table in the text
\end{table}

\medskip
\noindent
 Using our methods in Section \ref{sec-ccz} and \ref{sec-degree}, we can reproduce the above result in a
unified way. Recall that, for a given function $F_0(x)=G(x)=x^{2^i+1}$, we constructed
 $$
F(x)=F_0(x)+cx, \quad
  \caliL(x,y)=(x+\gamma\Tr(\alpha x +\beta y), y), \quad H(x)=x+\gamma\Tr(\alpha x+\beta F(x))
 $$
 satisfying
 \begin{align*}
     \Tr(\alpha\gamma)&=-2=0, \\
     \Tr(c\beta\gamma)&= \Tr(\beta D_\gamma F_0(0))= \Tr(\beta D_\gamma F_0(x)) \quad \text{ for all } x\in \Fbn,
 \end{align*}
 where the second condition guarantees that $\gamma$ is a $0$-linear structure of $\Tr(\beta F(x))$ (See Proposition
 \ref{key-prop}). Under these conditions, we proved that $F_0 \eqccz F\circ H$.
   For simplicity, one may first choose $\alpha=0$.  Also, choosing $\gamma=1=\beta$, one has
   $\Tr(\beta D_\gamma F_0(x))=\Tr((x+1)^{2^i+1}+x^{2^i+1})=\Tr(x^{2^i}+x+1)=\Tr(1)=n$ for all $x\in\Fbn$.
   Therefore it is enough to choose $c$ satisfying
   $\Tr(c)=\Tr(c\beta\gamma)=n=\Tr(1)$. When $n$ is odd, then one may choose $c=1$ so that $\Tr(c)=1$.
   In this case, one has
   $$F(x)=x^{2^i+1}+x, \quad H(x)=x+\gamma\Tr(\alpha x +\beta F(x))=x+\Tr(x^{2^i+1}+x)=x+\eps$$ and
   $$F\circ H(x)=F(x)+\eps D_\gamma F(x)=x^{2^i+1}+x+\Tr(x^{2^i+1}+x)(x^{2^i}+x),$$
   which is EA-equivalent to $F'(x)$ in \cite{BCP06}.
   On the other hand, when $n$ is even, one may choose either $c=0$ or $c=1$ since $\Tr(c)=0=\Tr(1)$
   is satisfied in either case. Therefore, if $c=0$, one gets
  $$F(x)=x^{2^i+1}, \quad H(x)=x+\gamma\Tr(\alpha x +\beta F(x))=x+\Tr(x^{2^i+1})=x+\eps$$ and
   $$F\circ H(x)=F(x)+\eps D_\gamma F(x)=x^{2^i+1}+\Tr(x^{2^i+1})(x^{2^i}+x+1),$$
   which is same to $F'(x)$ in \cite{BCP06}. Also if $c=1$, then one has
    $$F(x)=x^{2^i+1}+x, \quad H(x)=x+\gamma\Tr(\alpha x +\beta F(x))=x+\Tr(x^{2^i+1}+x)=x+\eps$$ and
   $$F\circ H(x)=F(x)+\eps D_\gamma F(x)=x^{2^i+1}+x+\Tr(x^{2^i+1}+x)(x^{2^i}+x),$$
   which is CCZ-equivalent to $x^{2^i+1}$. In other words, one concludes that
   \begin{remark}\label{goldremark}
   Independent of the parity of $n$, one has
   \begin{align*}
        x^{2^i+1} \noteqea  x^{2^i+1}+(x^{2^i}+x) \Tr(x^{2^i+1}+x) \eqccz x^{2^i+1}                                           .
   \end{align*}
   \end{remark}

 %\begin{remark}
 %%  {\color{red} Please prove the following claim: For $\gcd(i,n)=1$, one has
 % \begin{align*}
 %         F'(x)=x^{2^i+1}+(x^{2^i}+x) \Tr(x^{2^i+1}+x) \noteqea x^{2^i+1}+(x^{2^i}+x+1) \Tr(x^{2^i+1})=F''(x).
 %  \end{align*}
 %  The case $n=$odd is trivial because $\NL(F'')=0$ ($\because \Tr(F'')=0$) while $F'$ is APN. I need a proof for the
 %  case $n=$even. Seems to be NOT Difficult but I cannot think of the proof. }
 %  \end{remark}

\subsection{Other APN functions on $\Fbn$}

One may also apply our main theorem to all quadratic APN functions found so far. For example, in page $407$ of the book
of Carlet \cite{Car20}, there is a table which lists  CCZ-inequivalent quadratic functions. We rewrite the list in our
Table \ref{maintable} with our choice of $\alpha,\beta,\gamma$ and $c$.

\begin{table}[h!]
     \scriptsize
    \caption{Known classes of quadratic APN functions CCZ inequivalent to power functions on $\Fbn$,
      and $\ell(x)$ satisfying $\ell(\beta)=0$}
    \label{maintable}
    \smallskip
    %\centering
    \renewcommand{\arraystretch}{1.5}
    \begin{tabular}{cccc}
         \hline
        $F_0(x)$ &  Conditions on $F_0$ &  $\ell(x)$  & $\begin{array}{l} \text{APN} \\ \text{proven} \end{array}$ \\ \hline \hline
%%%%%%%%%%%%%%%%%%%%%%%%%%%%%%%%%%%%%%%%%%%%%%%%%%%%%%%%%%%%%%%%%%%%%%%%%%%
%1.
$\begin{array}{l}x^{2^{s}+1}+g^{2^{k}-1} x^{2^{i k}+2^{tk+s}}\end{array}$ &$\begin{array}{l}
    n=p k, \operatorname{gcd}(k, p)=\operatorname{gcd}(s, p k)=1, \\
    p \in\{3,4\}, i=s k \bmod p, t=p-i, \\
    n \geq 12, g \text { primitive in } \mathbb{F}_{2^{n}}^\times
\end{array}$
&$\begin{array}{l}
    x^{2^{n-s}} +x+\\
    (\omega x)^{2^{tk}} + (\omega x)^{2^{ik-s}}
\end{array}$
& \cite{BCL08} \\ \hline
%%%%%%%%%%%%%%%%%%%%%%%%%%%%%%%%%%%%%%%%%%%%%%%%%%%%%%%%%%%%%%%%%%%%%%%%%%
%2.
$\begin{array}{l}
        a x^{2^{i}(q+1)}+x^{2^{i}+1}+ \\
        x^{q\left(2^{i}+1\right)}+x^{q+1}+\\
        v x^{2^{i} q+1}+v^{q} x^{2^{i}+q}
\end{array}$
&$\begin{array}{l}
    q=2^{m}, n=2 m, \operatorname{gcd}(i, m)=1, \\
    v \in \mathbb{F}_{2^{n}}, a \in \mathbb{F}_{2^{n}} \backslash \mathbb{F}_{q}, \\
    X^{2^{i}+1}+v X^{2^{i}}+v^{q} X+1\\
    \text{has no zero in } \mathbb{F}_{2^n} \text{ s.t. } x^{q+1}=1
\end{array}$
&$\begin{array}{l}
    \left[(a+v^{2^m}+1)x\right] ^{2^{2m-i}}+\\
    \left[(a+v+1)x\right] ^{2^{m-i}} +\\
    v x^{2^m}+v x
\end{array}$
& \cite{BC08} \\ \hline
%%%%%%%%%%%%%%%%%%%%%%%%%%%%%%%%%%%%%%%%%%%%%%%%%%%%%%%%%%%%%%%%%%%%%%%%%%
%3.
$\begin{array}{l}
    x^{3}+a^{-1} \Tr\left(a^{3} x^{9}\right)
\end{array}$
&$\begin{array}{l}
    a \neq0
\end{array}$
&$\begin{array}{l}
    x^{2^{n-1}} + x +  \tilde{a}\Tr(\frac{x}{a})\\
\end{array}$
&\cite{BCL09FFTA} \\ \hline
%%%%%%%%%%%%%%%%%%%%%%%%%%%%%%%%%%%%%%%%%%%%%%%%%%%%%%%%%%%%%%%%%%%%%%%%%%
%4.
$\begin{array}{l}
    x^{3}+a^{-1} \Tr_{3}^{n}\left(a^{3} x^{9}+a^{6} x^{18}\right)
\end{array}$
&$\begin{array}{l}
    3\mid n, a \neq 0
\end{array}$
&$\begin{array}{l}
     x^{2^{n-1}} + x + \tilde{a}\Tr^n_3\left(\frac{x}{a} + \left(\frac{x}{a}\right)^4\right)\\
\end{array}$
& \cite{BCL09IEEE} \\ \hline
%%%%%%%%%%%%%%%%%%%%%%%%%%%%%%%%%%%%%%%%%%%%%%%%%%%%%%%%%%%%%%%%%%%%%%%%%%
%5.
$\begin{array}{l}
    x^{3}+a^{-1} \Tr_{3}^{n}\left(a^{6} x^{18}+a^{12} x^{36}\right)
\end{array}$
&$\begin{array}{l} 3\mid n, a  \neq 0
\end{array}$
&$\begin{array}{l}
     x^{2^{n-1}} + x + \tilde{a}\Tr^n_3\left(\left(\frac{x}{a}\right)^2 + \left(\frac{x}{a}\right)^4\right)\\
\end{array}$
&\cite{BCL09IEEE} \\ \hline
%%%%%%%%%%%%%%%%%%%%%%%%%%%%%%%%%%%%%%%%%%%%%%%%%%%%%%%%%%%%%%%%%%%%%%%%%%
%6.
$\begin{array}{l}
    g x^{2^{s}+1}+g^{2^{k}} x^{2^{2 k}+2^{k+s}}+ \\
    v x^{2^{2 k}+1}+w g^{2^{k}+1} x^{2^{s}+2^{k+s}}
\end{array}$
&$\begin{array}{l}
    n=3 k, \operatorname{gcd}(k, 3)=\operatorname{gcd}(s, 3 k)=1, \\
    v, w \in \mathbb{F}_{2^{k}}, v w \neq 1, \\
    3 \mid(k+s),  g \text { primitive in } \mathbb{F}_{2^{n}}^\times
\end{array}$
&$\begin{array}{l}
\left[(g + w g^{2^{k}+1})x\right]^{2^{3k-s}}\hspace{-1em}+\\
\left[(g^{2^{k}} + w g^{2^{k}+1})x\right]^{2^{2k-s}} \hspace{-1em}+\\
\left[(g^{2^{k}} +  v)x\right]^{2^{k}} +    (g+v) x
\end{array}$
&\cite{BBMM08} \\ \hline
%%%%%%%%%%%%%%%%%%%%%%%%%%%%%%%%%%%%%%%%%%%%%%%%%%%%%%%%%%%%%%%%%%%%%%%%%%
%7.
$\begin{array}{l}
    \left(x+x^{2^{m}}\right)^{2^{k}+1} + \\
    v \left(g x+g^{2^{m}} x^{2^{m}}\right)^{ (2^{k}+1) 2^{i}}+ \\
    g\left(x+x^{2^{m}}\right)\left(g x+g^{2^{m}} x^{2^{m}}\right)
\end{array}$
&$\begin{array}{l}
    n=2 m, m \geq 2 \text { even, } \\
    \operatorname{gcd}(k, m)=1 \text { and } i \geq 2 \text { even }\\
    g \text{ primitive in } \mathbb{F}_{2^{n}}^\times, v \in \mathbb{F}_{2^{m}} \text{ not cube}
\end{array}$
&$\begin{array}{l}
g\tilde{g}^{ 2^{m-k} }\left[(vx)^{2^m} + vx\right]^{ 2^{m-k-i} } \hspace{-2em}+\\
g\tilde{g}^{ 2^{k} } \left[(vx)^{2^{m}}+vx\right]^{2^{m-i}} \hspace{-1em}+\\
\tilde{g}\left[(gx)^{2^{m}} + gx\right]
\end{array}$
& \cite{ZP13}\\ \hline
%%%%%%%%%%%%%%%%%%%%%%%%%%%%%%%%%%%%%%%%%%%%%%%%%%%%%%%%%%%%%%%%%%%%%%%%%%
%8.
$\begin{array}{l}
    a^{2} x^{2^{2 m+1}+1}+v^{2} x^{2^{m+1}+1}+\\
    a x^{2^{2 m}+2} + v x^{2^{m}+2}+\\
    \left(w^{2}+w\right) x^{3}
\end{array}$    &
$\begin{array}{l}n=3 m, m \text{ odd; } U \text{ subgroup of}\\
\mathbb{F}_{2^{n}}^\times \text { of order } 2^{2 m}+2^{m}+1,\\
L(x)= a x^{2^{2 m}}+v x^{2^{m}}+w x \in \mathbb{F}_{2^{n}}[x], \\
\text { s.t. }\forall u_1, u_2 \in U, L(u_1) \notin\{0, u_1\} \text { and } \\
u_1^{2} L(u_2) +u_2 L(u_1)^{2} \neq 0 \\
\Rightarrow \frac{u_2^{2} L(u_1)+u_1 L(u_2)^{2}}{u_1^{2} L(u_2)+u_2 L(u_1)^{2}} \notin \mathbb{F}_{2^{m}} \\
\end{array}$ & $\begin{array}{l}
\left[\left(a  + v  + w^{2}+w\right)x\right]^{2^{3m-1}}  +\\
\left(v x \right)^{2^{2m}} +
\left(v^2 x\right)^{2^{2m-1}} +\\
\left(a x \right)^{2^{m}} +
\left(a^2 x\right)^{2^{m-1}} +\\
\left(a^{2} + v^{2}+ w^{2}+w \right)x
\end{array}$& \cite{BCCCV} \\ \hline
%%%%%%%%%%%%%%%%%%%%%%%%%%%%%%%%%%%%%%%%%%%%%%%%%%%%%%%%%%%%%%%%%%%%%%%%%%
%9.
$\begin{array}{l}
    v\left(v^{q} x+v x^{q}\right)\left(x^{q}+x\right)+\\
    \left(v^{q} x+v x^{q}\right)^{2^{2 i}+2^{3 i}}+  \\
    a\left(v^{q} x+v x^{q}\right)^{2^{2 i}}\left(x^{q}+x\right)^{2^{i}}+\\
    w\left(x^{q}+x\right)^{2^{i}+1}
\end{array}$        & $\begin{array}{l}
q=2^{m}, n=2 m, \operatorname{gcd}(i, m)=1 \\
X^{2^{i}+1}+a X+w \\
\text { has no zero in } \mathbb{F}_{2^{m}}
\end{array}$
&$\begin{array}{l}
\tilde{v}^{2^{i}}\left[(ax)^{2^m}+ ax\right]^{2^{m-i}}+ \\
v^{2^{m}}\tilde{v}^{2^{i}} \left[x^{2^{m}}+ x\right]^{2^{m -2i}} + \\
v^{2^{m}} \tilde{v}^{2^{m-i}}  \left[x^{2^{m}}+x \right]^{2^{m - 3 i}}+ \\
\tilde{v} \left[ (vx)^{2^m}+vx \right]
\end{array}$& \cite{Tani19} \\ \hline
%%%%%%%%%%%%%%%%%%%%%%%%%%%%%%%%%%%%%%%%%%%%%%%%%%%%%%%%%%%%%%%%%%%%%%%%%%
%10.
$\begin{array}{l}
    x^{3}+\omega\left(x^{2^{i}+1}\right)^{2^{k}}+\\
    \omega^{2} x^{3 \cdot 2^{m}}+    \left(x^{2^{i+m}+2^{m}}\right)^{2^{k}}
\end{array}$        &$\begin{array}{l}n=2 m ; m \text { odd } ; 3 \nmid m \\
i=m-2 \text { or } i \equiv (m-2)^{-1}\pmod{n}\\
\omega \in \mathbb{F}_{4} \setminus \mathbb{F}_2\end{array}$  &
$\begin{array}{l} x^{2^{2m-1}}   +\omega^{2^k} x^{2^{2m-k}}+\\
\omega^{2^{k+1}} x^{2^{2m- k-i}} +\omega x^{2^m}+\\
\omega^2 x^{ 2^{m-1}}  + x^{2^{m-k}}  +\\
x^{2^{m-k-i}}+ x\end{array}$& \cite{BHK19}\\ \hline
    \end{tabular}

\medskip
\noindent $-$ $\beta \in \Fbn$ is a unique nonzero root of $\ell(x)$.\\
\noindent $-$ $\tilde{a} = a^{3\cdot 2^{n-3}} + a^{3} $,  $\tilde{g} =  g^{2^m} + g$, and $\tilde{v} = v^q + v$ in the
third column.

\end{table}
%\noindent $-$ $\beta$ is a unique nonzero root of $\ell(x)$.

\normalsize

\bigskip

\noindent {\bf Explanation of Table \ref{maintable}:}
\medskip

\noindent
For all APN function $F_0(x)$ in the Table \ref{maintable}, we choose
\begin{align}
\alpha=0, \gamma= 1 \text{ and } c=F_0(1), \label{table-parameter}
\end{align}
and $\beta\in \Fbn$ is a unique nonzero root of linearized $\ell(x)$. Then the resulting function
 \begin{align*}
 F\circ H(x)&=F(x)+\Tr(\beta F(x))D_1F(x) \\
 &=F_0(x)+cx +\Tr(\beta F(x))(F_0(x+1)+F_0(x)+c)
 \end{align*}
is cubic and is CCZ-equivalent to $F_0(x)$. Here, $\ell(x)$ is determined as follows:
 Since $\beta$ must satisfy $\Tr(\beta L_1(x))=0$ for all $x$ with $L_1(x)=F_0(x+1)+F_0(x)+F_0(1)+F_0(0)$,
 we write $L_1(x)=\sum_{i=0}^{n-1}a_ix^{2^i}$. Then
 \begin{align*}
 0=\Tr(\beta L_1(x))=\Tr\left(\sum_{i=0}^{n-1}\beta a_i x^{2^i}\right)=
          \Tr\left(\left(\sum_{i=0}^{n-1}(\beta a_i)^{2^{n-i}}\right) x\right).
 \end{align*}
 Therefore, $\beta$ is a root of
 \begin{align}
   \ell(x)=\sum_{i=0}^{n-1} a_{n-i}^{2^i}x^{2^i}. \label{ell}
 \end{align}

\bigskip

\noindent The choices of $\alpha, \beta, \gamma $ and $c$ are not unique and there are many alternative choices are
possible. For example, other than Gold functions, Budaghyan et al. \cite{BCL09FFTA} found CCZ-equivalent $F'$ of
algebraic degree $3$ and $4$ for the quadratic $F_0(x)=x^3+a^{-1}\Tr(a^3x^9)$ with different parameters. Among the $10$
APN functions in Table \ref{maintable}, $F_0(x)=x^3+a^{-1}\Tr(a^3x^9)$ is the only (published) one where the existence
of (non-obvious) $F'$  is known such that $F'\eqccz F_0$ and $F'\noteqea F_0$. Moreover, (quadratic) APN permutation is
very rare in the sense that there is no quadratic APN permutation on even dimension \cite{SZZ94}, and for odd $n$, the
only permutation in Table \ref{maintable} is the function in the first row
 with odd $n=3k$ and $\gcd(k,3)=1=\gcd(s,3k)$. It is conjectured in \cite{BCL08} that the algebraic degree of the inverse
 function is $\frac{3k+1}{2}$. Therefore, based on that conjecture, it is very probable that our cubic $F\circ H$ is
 also EA-inequivalent to the inverse function.

\medskip

It should be mentioned that, though not listed in Table \ref{maintable}, there is a sporadic case of (recently
discovered)  quadratic APN permutation on dimension $9$. That is, Beierle and Leander \cite{BL22} found  the following
two quadratic APN permutations on $\F_{2^9}$:
 \begin{align}
   F_0(x)&=x^3+u^2x^{2^3+2}+ux^{2^4+2^3}+u^4x^{2^6+2^4}+u^6x^{2^7+2^3}, \label{spo1} \\
   F_1(x)&=x^3+ux^{2^3+2}+u^2x^{2^4+1}+u^4x^{2^6+2^4}+u^5x^{2^7+2^6}, \label{spo2}
 \end{align}
 where $u\in \F_{2^3}$ with $u^3+u+1=0$. (See also \cite{BCLP22}.) It is straightforward to check that the inverses of $F_0$ and $F_1$ both
 have algebraic degree $5$. Therefore our $F\circ H$ are  non-obvious functions which are CCZ-equivalent but not
 EA-equivalent to the functions in Equations  \eqref{spo1} and \eqref{spo2}. Let us explain how the `Explanation of Table \ref{maintable}'
  can also be applied to this sporadic case.
 \begin{example}
  Let
  $$F_0(x)=x^3+u^2x^{2^3+2}+ux^{2^4+2^3}+u^4x^{2^6+2^4}+u^6x^{2^7+2^3} \in \F_{2^9}[x]$$
   be a APN permutation in Equation
   \eqref{spo1}. Let $\alpha=0, \gamma=1$ and $c=F_0(1)=u^2$ similarly as in Table \ref{maintable} (See Equation \eqref{table-parameter}.).
   Then, one gets
   \begin{align*}
     L_1(x)&=F_0(x+1)+F_0(x)+F_0(1)+F_0(0) \\
          &=x^2+x+u^2(x^{2^3}+x^2)+u(x^{2^4}+x^{2^3})+u^4(x^{2^6}+x^{2^4})+u^6(x^{2^7}+x^{2^3})\\
          &=x+u^6x^2+u^3x^{2^3}+u^2x^{2^4}+u^4x^{2^6}+u^6x^{2^7}.
   \end{align*}
   The corresponding $\ell(x)$ (see Equation \eqref{ell}) is easily calculated as
   $$
   \ell(x)=u^3x^{2^8}+u^3x^{2^6}+ux^{2^5}+u^4x^{2^3}+u^3x^{2^2}+x,
   $$
   and $\beta=u^5=u^2+u+1$ is a unique nonzero root of $\ell(x)$. Therefore one has
  \begin{align*}
   F(x)=F_0(x)+u^2x,\,\, H(x)=x+\gamma\Tr(\alpha x +\beta F(x))=x+\Tr(u^5F_0(x)+x),
   \end{align*}
   and
   \begin{align*}
  F\circ H(x)&= F(x)+\Tr(\beta F(x))D_1 F(x) \\
          &=F_0(x)+u^2x+\left(F_0(x+1)+F_0(x)+u^2\right)\Tr\left(u^5F_0(x)+x\right),
   \end{align*}
   where
   \begin{align*}
    F\circ H \eqccz F_0  \quad \text{and} \,\,F\circ H \noteqea F_0.
   \end{align*}
   The case for $F_1$ in Equation \eqref{spo2} can be dealt in the same manner.
   \end{example}

\bigskip

 We present another example from the Table \ref{maintable},
 where one may use different approach for determining $\alpha, \beta,
\gamma$ and $c$.

\begin{example}
Let
\begin{align*}
F_0(x)=x^3+a^{-1}\Tr_3^n(a^3x^9+a^6x^{18}) \quad \text{ with }\, 3|n \text { and } a\in \Fbn^\times.
\end{align*}
  It is well-known \cite{BCL09IEEE} that
 $F_0$ is APN for all nonzero $a$. Though a systematic approach is possible as is shown in Table \ref{maintable},
 we provide an alternative description on finding CCZ-equivalent $F\circ H$ for this case.
 %We choose this example because the expression of the linearized $\ell(x)$ is rather complicated for this case due to the
 %$\Tr_3^n$ term in $F_0$.
 For this $F_0(x)$, letting $q=2^3$, one has
 \begin{align}
 L_\gamma(x)=x^2\gamma+x\gamma^2+a^{-1}J \quad \text{ with }\,
   J=\Tr_3^n(a^3(x^q\gamma+x\gamma^q)+a^6(x^{2q}\gamma^2+x^2\gamma^{2q})) \in \F_{2^3}. \label{eqeg-nonGold}
 \end{align}

\noindent {\it When $\Tr_3^n(a)=0$:}

\noindent Let $\beta=a^3, \gamma=a^{-1}$ and $c=a^{-2}$. Then, for all $x\in \Fbn$, one gets
\begin{align*}
 \Tr(\beta L_\gamma(x))&=\Tr(x^2\beta\gamma+x\beta\gamma^2)+\Tr_1^3\Tr_3^n(a^2J)\\
    &=\Tr(x^2\beta\gamma+x\beta\gamma^2)+\Tr_1^3 (J[\Tr_3^n(a)]^2)
    =\Tr(x^2\beta\gamma+x^2\beta^2\gamma^4)  \quad (\because \Tr_3^n(a)=0)\\
    &=\Tr(\beta\gamma(1+\beta\gamma^3)x^2)=0. \quad\quad (\because \beta\gamma^3=1)
\end{align*}
Also, from
 \begin{align*}
 \Tr(\beta(F_0(\gamma)+F_0(0)))&=\Tr(a^3\{a^{-3}+a^{-1}\Tr_3^n(a^3\gamma^9+a^6\gamma^{18}) \}) \\
   &=\Tr(1)+ \Tr_1^3\Tr_3^n\{a^2\Tr_3^n(a^3\gamma^9+a^6\gamma^{18}) \}=\Tr(1), \quad (\because \Tr_3^n(a)=0)
 \end{align*}
 one has
 $$
 \Tr(c\beta\gamma)=\Tr(1)=\Tr(\beta(F_0(\gamma)+F_0(0))).
 $$

\smallskip
 \noindent {\it When $\Tr_3^n(a)\neq 0$: In this case, let $\omega=\Tr_3^n(a) \in \F_{2^3}^\times$.}

\noindent Let $\beta=\omega^5 a^3, \gamma=\omega^3 a^{-1}$ and $c=\omega^6 a^{-2}$. Then, from Equation
\eqref{eqeg-nonGold},
 one gets
\begin{align*}
 \Tr(\beta L_\gamma(x))&=\Tr(x^2\beta\gamma+x\beta\gamma^2)+\Tr_1^3\Tr_3^n(\omega^5a^2J)\\
    &=\Tr(x^2\beta\gamma+x\beta\gamma^2)+\Tr_1^3 (J\omega^5[\Tr_3^n(a)]^2) \\
    &=\Tr(x^2\beta\gamma+x^2\beta^2\gamma^4)+\Tr_1^3(J)  \quad (\because \Tr_3^n(a)=\omega \text{ and } \omega^7=1).
   % &=\Tr(\beta\gamma(1+\beta\gamma^3)x^2)=0. \quad\quad (\because \beta\gamma^3=1)
\end{align*}
Therefore, from Equation \eqref{eqeg-nonGold}, writing $J=\Tr_3^n(y+y^2)$ with $y=a^3(x^q\gamma+x\gamma^q)$,
 one finally has
$$
 \Tr(\beta L_\gamma(x))=\Tr(\beta\gamma(1+\beta\gamma^3)x^2)+\Tr_1^3\Tr_3^n(y+y^2)=0.
  \quad (\because \beta\gamma^3=1 \text{ and } \Tr_1^n(y+y^2)=0)
$$

 \noindent
 Also, from
 \begin{align*}
 \Tr(\beta(F_0(\gamma)+F_0(0)))&=\Tr(\omega^5a^3\{\omega^9a^{-3}+a^{-1}\Tr_3^n(a^3\gamma^9+a^6\gamma^{18}) \}) \\
   &=\Tr(1)+ \Tr_1^3\Tr_3^n\{\omega^5a^2\Tr_3^n(a^3\gamma^9+a^6\gamma^{18}) \} \\
   &=\Tr(1)+ \Tr_1^3\Tr_3^n(a^3\gamma^9+a^6\gamma^{18} )=\Tr(1), \quad (\because \Tr_3^n(a)=\omega \text{ and } \omega^7=1)
 \end{align*}
 one has
 $$
 \Tr(c\beta\gamma)=\Tr(1)=\Tr(\beta(F_0(\gamma)+F_0(0))).
 $$

\noindent
 As a consequence, assuming $a\neq 0$ and $3|n $, one gets
\begin{align*}
 x^3+a^{-1}\Tr_3^n(a^3x^9+a^6x^{18})= F_0(x) \eqccz F\circ H(x)= F(x)+\Tr(\alpha x+\beta F(x))D_\gamma F(x),
\end{align*}
and $F_0\noteqea F\circ H$,
  where $F(x)=F_0(x)+cx$ and $\alpha,\beta, \gamma$ and $c$ are chosen as in Table \ref{budaegtrace}.

\begin{table}[ht]
\caption{$\alpha,\beta, \gamma$ and $c$ satisfying $F\circ H \eqccz F_0$ but $F\circ H \noteqea F_0$}
 \label{budaegtrace}
\smallskip
\centering
 \renewcommand{\arraystretch}{1.1}
\begin{tabular}{c || c | c | c | c}  % centered columns (4 columns)
\hline\hline
 & $\alpha$ & $\beta$ & $\gamma$ & $c$ \\ [0.5ex] % inserts table
%heading
\hline\hline
 When $\Tr_3^n(a)=0$ & $0$ & $a^3$ & $a^{-1}$ & $a^{-2}$ \\
\hline
 When $\Tr_3^n(a)=\omega\neq 0$ & $0$ & $\omega^5a^3$ & $\omega^3a^{-1}$ & $\omega^6 a^{-2}$\\
\hline
\end{tabular}
\label{table-egnongold}  % is used to refer this table in the text
\end{table}
\end{example}

\begin{remark}
The above choice of $\alpha,\beta, \gamma$ and $c$ is just one example of parameter choices.
 For APN case as above, for any $\gamma\neq 0$, there exists a unique $\beta\neq 0$ satisfying
 $\Tr(\beta L_\gamma(x))=0$ for all $x$. Also, for each given choice of $(\gamma,\beta)$ pair,
 there are $2^{n-1}$ possible choices of $\alpha$ and $c$, respectively. For example, one may choose
 any $\alpha\in \Fbn$ satisfying $\Tr(\alpha\gamma)=0$, and $\alpha=0$ or $\alpha=(\omega^{-1}+\omega^{-2})a$
 is just one of such choices.
\end{remark}

\medskip

\subsection{Gold Functions for $p>2$}\label{sec-paryGold}

For Gold function $G(x)=x^{p^i+1} \in \Fpn[x]$ with $p$ odd, the following is well-known.
 (See for example \cite{Zie15, BT20}.)
 \begin{enumerate}
 \item[1.] $G(x)$ is planar (i.e., PN) if and only if $\frac{n}{\gcd(i,n)}$ is odd.
 \item[2.] If $\frac{n}{\gcd(i,n)}$ is even, then $|\Ker L_\gamma |= p^{\gcd(i,n)}$ for every $\gamma\in \Fpn^\times$,
   where $L_\gamma(x)=G(x+\gamma)-G(x)-(G(\gamma)-G(0))=\gamma x^{p^i}+\gamma^{p^i}x$. Consequently, one has
   $\Delta_G=p^{\gcd(i,n)}$.
 \end{enumerate}

\begin{remark}
 It is proven in \cite{BH08, KP08}  that for the case of planar function $F(x)$, one has
 $\mathscr{C}_\textrm{EA}(F)=\mathscr{C}_\textrm{CCZ}(F)$, i.e., every $F'$ satisfying $F'\eqccz F$ also satisfies
  $F'\eqea F$.
\end{remark}

\noindent From the previous remark, we are only interested in the case of $G(x)=x^{p^i+1}\in \Fpn[x]$ with even $n$
 such that $\frac{n}{\gcd(i,n)}$ is also even.

%\begin{remark}
%{\color{red} ??? We may delete this} When $\frac{n}{d}$ is even where $d=\gcd(i,n)$, it is shown in \cite{Cou981,
%Cou982} (See also \cite{XLZT23}.)
% that $|\cW_G(a,b)| \in \{0, p^\frac{n}{2},p^{\frac{n}{2}+d}\}$ for all $a,b\in \Fpn$ such
% that $|\cW_G|=p^{\frac{n}{2}+d}$. On the other hand, our function $F\circ H$ satisfying
% $F\circ H\eqccz F\eqea F_0=G$ has an algebraic degree $4$ if $|\cW_G|\leq p^{n-1}$ due to Proposition
% \ref{quartic-degree}. Therefore one has $F\circ H \not \eqccz G$ if
% $\frac{n}{2}+d\leq n-1$ (i.e., if $d\leq \frac{n}{2}-1$).
%Since any proper divisor $d$ of $n$ must satisfy $d\leq \frac{n}{2}$,
% it is guaranteed that
% one always have $d^\circ (F\circ H)=4$ (i.e., $F\circ H \not\eqccz G$) if $i\neq \frac{n}{2}$ and $\frac{n}{d}=\text{even}$.
%\end{remark}

\begin{corollary}\label{coro-quartic}
Let $n\geq 4$ and $G(x)=x^{p^i+1} \in \Fpn[x]$ where $\frac{n}{d}$ is even with $d=\gcd(i,n)$. Suppose that $i\neq
\frac{n}{2}$. Then the function in Theorem \ref{mainthm}
 \begin{align*}
  F\circ H(x) &=F(x)+\eps D_\gamma F(x) +F_{\text{\tiny DO}}(\gamma)(\eps^2-\eps) \\
            &=F(x)+\eps D_\gamma F(x) +\gamma^{p^i+1}(\eps^2-\eps) \\
            (\text{with } &F(x)=x^{p^i+1}+cx, \,\,\, \eps=\Tr(\alpha x +\beta F(x)), \,\,\, H(x)=x+\gamma\eps)
 \end{align*}
 is of algebraic degree $4$ and  $F\circ H\eqccz F\eqea F_0=G$.
\end{corollary}

\noindent
\begin{proof}
When $\frac{n}{d}$ is even with $d=\gcd(i,n)$, it is shown in \cite{Cou981, Cou982} (See also \cite{XLZT23}.)
 that $|\cW_G(a,b)| \in \{0, p^\frac{n}{2},p^{\frac{n}{2}+d}\}$ for all $a,b\in \Fpn$, which implies
 that $|\cW_G|=p^{\frac{n}{2}+d}$. On the other hand, from Proposition
 \ref{quartic-degree}, our function $F\circ H$ satisfying
 $F\circ H\eqccz F\eqea F_0=G$ has an algebraic degree $4$ if $|\cW_G|\leq p^{n-1}$.
  Therefore one has $F\circ H \not \eqea G$ if
 $\frac{n}{2}+d\leq n-1$ (i.e., if $d\leq \frac{n}{2}-1$).
Since any proper divisor $d$ of $n$ must satisfy $d\leq \frac{n}{2}$,
 it is guaranteed that
 one always have $d^\circ (F\circ H)=4$ (i.e., $F\circ H \noteqea G$) if $i\neq \frac{n}{2}$ and $\frac{n}{d}=\text{even}$.
\end{proof}

\medskip
\noindent For explicit choices of $\alpha,\beta,\gamma$ and $c$, let us first choose $\gamma=1$. Then, from Proposition
\ref{key-prop}, $\alpha$ is determined by the condition $\Tr(\alpha\gamma)=\Tr(\alpha)=-2$. Also, since
$L_\gamma(x)=x^{p^i}\gamma+\gamma^{p^i}x=x^{p^i}+x$, one gets
$$
 \Tr(\beta L_\gamma(x))=\Tr(\beta (x^{p^i}+x))=\Tr((\beta^{p^{n-i}}+\beta)x)=0 \quad \text{ for all } x\in \Fpn
$$
if and only if $\beta^{p^{n-i}}=-\beta$, i.e., $\beta^{p^{n-i}-1}=-1$. Letting $n=dn'$ and $i=di'$ with $d=\gcd(n,i)$,
one has $n'=\text{even}$ and $i'=\text{odd}$ $(\because \frac{n}{d}=\text{even})$. Since
 $$
 \beta^{p^{n-i}-1}=\left(\beta^{p^d-1}\right)^{\frac{(p^d)^{n'-i'}-1}{p^d-1}}
   =\left(\beta^{p^d-1}\right)^{(p^d)^{n'-i'-1}+(p^d)^{n'-i'-2}+\cdots +(p^d)+1},
 $$
 and since
 $$
 (p^d)^{n'-i'-1}+(p^d)^{n'-i'-2}+\cdots +(p^d)+1\equiv n'-i' \equiv 1 \pmod{2},
 $$
choosing $\beta$ satisfying $\beta^{p^d-1}=-1$ (i.e., $\beta^{p^d}+\beta=0$), one has
$$
 \beta^{p^{n-i}-1}=\left(\beta^{p^d-1}\right)^{\frac{(p^d)^{n'-i'}-1}{p^d-1}}=(-1)^{n'-i'}=-1.
$$
For this choice of $\beta$, one also gets
\begin{align*}
\Tr(\beta)=\Tr_1^d\Tr_d^n(\beta)
 &=\Tr_1^d \left(\beta+\beta^{p^d}+\beta^{p^{2d}}+\beta^{p^{3d}}+\cdots +\beta^{p^{(n'-2)d}}+\beta^{p^{(n'-1)d}}\right)\\
       & =\Tr_1^d \left(\beta+\beta^{p^d}+(\beta+\beta^{p^d})^{p^{2d}}+\cdots +(\beta+\beta^{p^d})^{p^{(n'-2)d}}\right)
        =\Tr_1^d(0)=0.
\end{align*}
Therefore, the condition on $c$ in Proposition \ref{key-prop},
 $$
 \Tr(c\beta\gamma)=-\Tr(\beta D_\gamma G(0))=-\Tr(\beta (G(\gamma)-G(0))=-\Tr(\beta\gamma^{p^i+1})=-\Tr(\beta)=0
$$
is satisfied by choosing $c=0$. Consequently, for our choice of parameters
  ($c=0, \gamma=1$, $\Tr(\alpha)=-2, \beta^{p^d}+\beta=0$), the function
  $F\circ H= G\circ H$ ($\because c=0$) is written as
 \begin{align*}
  G\circ H(x)
            &=G(x)+\eps D_\gamma G(x) +\gamma^{p^i+1}(\eps^2-\eps) \\
            &=x^{p^i+1}+\eps (x^{p^i}+x+1)+\eps^2-\eps,
            \qquad (\text{with } \eps=\Tr(\alpha x +\beta x^{p^i+1}) )
 \end{align*}
 where the algebraic degree of $G\circ H$ is $4$.

\begin{remark}\label{rmkalpbet}
The elements $\alpha$ and $\beta$ satisfying $\Tr(\alpha)=-2$ and $\beta^{p^d}+\beta=0$ can be chosen as follows.
 \begin{enumerate}
 \item[1.] If $\gcd(n,p)=1$, then $\Tr(1)=n\neq 0$. Thus letting
 $\alpha=-\frac{2}{n}\in \Fp$, one gets $\Tr(\alpha)=-2$. If $p|n$ or or if no restriction of $p$ and $n$ is given,
 let $\tau$ be any primitive root of $\Fpn^\times$ and let $h(x)$ be the monic irreducible polynomial of $\tau$. Then
 letting $\frac{h(x)}{x-\tau}=\sum_{i=0}^{n-1} a_i x^i $, one has $\Tr\left(\frac{a_0}{h'(\tau)}\right)=1$ where
  $h'(x)$ is the usual (formal) derivative of $h$. (See Theorem 1.5 and 4.23 in \cite{Men93}.) Therefore letting
   $\alpha=-2\frac{a_0}{h'(\tau)}$, one gets $\Tr(\alpha)=-2$.

 \item[2.] Choose any $\zeta \in \F_{p^{2d}}\setminus \F_{p^{d}}$. Since $2d | n$, one gets $\zeta\in \Fpn^\times$.
  Letting  $\beta=\zeta-\zeta^{p^d}$, one has $\beta^{p^d}+\beta=\zeta^{p^d}-\zeta+\zeta-\zeta^{p^d}=0$.

 \end{enumerate}

\end{remark}

\begin{corollary}\label{coro-quartic2}
Under the same assumption as in Corollary \ref{coro-quartic} and letting $c=0, \gamma=1$ and choosing $\alpha$ and
$\beta$ as in Remark \ref{rmkalpbet}, the function
 $$
 G\circ H(x)
            =x^{p^i+1}+\eps (x^{p^i}+x+1)+\eps^2-\eps
            \qquad (\text{with } \eps=\Tr(\alpha x +\beta x^{p^i+1}) )
 $$
  is CCZ-equivalent to $G(x)=x^{p^i+1}$ but not EA-equivalent to each other. That is,
 $$
  G\circ H \noteqea G \qquad \text{and } \quad  G\circ H \eqccz G.
 $$
 Moreover, $G\circ H$ is EA-inequivalent to any power function.
 \end{corollary}

\noindent
\begin{proof}
 Since  $G_{\text{\tiny DO}}(\gamma)= \gamma^{p^i+1}=1\neq 0$ and since $\Tr(c\beta\gamma)=0$, using
 Theorem \ref{mainthm}, one has $G\circ H \noteqea x^s$ for any $s\in \Z$.
\end{proof}

 %{\color{red}
%\begin{remark}
%These are the list of SAGE and other questions;
% \begin{enumerate}
% \item[1] Verify that absolute value $|\cW_F|$ for odd $p$ may not be an integer. (i.e., $\sqrt{p}$ may appear.) For example, compute
%  $|\cW_F(a,b)|$ when $G(x)=x^{3^i+1}$ over $\F_{3^n}$ for $n=3,4,5,6$.
% \item[2] from previous experiments, find $|\cW_F|$ (max abs of Walsh coeff) when $\Delta_G=3$. or for $G(x)=x^{p^i+1}$
% with $\Delta_G=p$.
%
% \end{enumerate}

%\end{remark}
% }

\section{Applications to Quadratic $(n,m)$-Functions}

In this section, we state that our main Theorem \ref{mainthm} in previous section has a natural generalization to
$(n,m)$-functions $: \Fpn \rightarrow \Fpm $, where $m$ is a positive divisor of $n$. One can define the Walsh
transform and the differential uniformity of $(n,m)$-function $F(x)$ in a similar manner. For a given quadratic
$(n,m)$-function $F$,  if $\Delta_F > p^{n-m}$, one can show that there exists $\gamma \in \Fpn^\times$ and $\beta\in
\Fpm^\times$ such that $\gamma$ is a linear structure of $\Tr_1^m(\beta F)$. That is, using the following
modifications;
$$
 \Delta_F>1 \longrightarrow  \Delta_F >p^{n-m} \quad \text{and}\quad \Tr(\beta F(x)) \longrightarrow \Tr_1^m(\beta F(x)),
$$
it is straightforward to show that Proposition \ref{LS-quad} also holds for the case of $(n,m)$-function $F$. Lemma
\ref{lemma-cczmap} is valid for the following $\mathscr{L}$:
\begin{align*}
 \mathscr{L} : \Fpn \times \Fpm &\rightarrow \Fpn \times \Fpm \\
(x,y) &\mapsto (x+\gamma (\Tr(\alpha x)+\Tr_1^m(\beta y)), y),
\end{align*}
where $\alpha,\gamma \in \Fpn$ and $\beta \in \Fpm$. Also, for  a given $(n,m)$-function $F_0$, Proposition
\ref{key-prop} is still valid  using the following modifications;
\begin{align*}
 F(x)=F_0(x)+cx \quad &\longrightarrow\quad  F(x)=F_0(x)+\Tr_m^n(cx), \\
 \Tr(c\beta\gamma)= -\Tr(\beta D_\gamma F_0(0)) \quad &\longrightarrow\quad \Tr(c\beta\gamma)= -\Tr_1^m(\beta D_\gamma
 F_0(0)), \\
  H(x)=x+\gamma\Tr\left(\alpha x +\beta F(x)\right) \quad &\longrightarrow\quad
  H(x)=x+\gamma \left( \Tr (\alpha x) +\Tr_1^m (\beta F(x)) \right).
\end{align*}
With these modifications, the validity of Theorem \ref{thm1} and Proposition \ref{key-prop2} is unchanged. Lemma
\ref{lemF3} is valid with the modifications $\Tr(\beta F(x)) \rightarrow \Tr_1^m(\beta F(x))$ and $\Tr(\beta
L_{\gamma_i}(\gamma_j)) \rightarrow \Tr_1^m(\beta L_{\gamma_i}(\gamma_j))$. Proposition \ref{cubic-degree} is still
true with no modification except for the previously mentioned $\cW_{\Tr_1^m(bF)}$ and $\eps=\Tr (\alpha x) +\Tr_1^m
(\beta F(x))$. In the exactly same manner as in Lemma \ref{lemF3}  and Proposition \ref{cubic-degree}, one finds that
Lemma \ref{quartic} and Proposition \ref{quartic-degree} are also valid. Summing up all these together, one has the
following generalization of Theorem \ref{mainthm} as follows.

\begin{theorem}{\bf(Generalization of Theorem \ref{mainthm})}\label{mainthm2}
 Let $F_0(x)$ be a quadratic $(n,m)$-function on $\Fpn$ with $n\geq 4$ and $m$ is a positive divisor of $n$.
 Then,  one concludes the
 followings.
 \begin{enumerate}
 \item[1.] When $p=2$: Suppose that $|\cW_{F_0}|<2^n$ (i.e., $\NL(F_0)\neq 0$) and $2^{n-m}<\Delta_{F_0}\leq 2^{n-3}$.
       Then for any $\gamma\in\Fbn^\times$, one can choose $\beta\in\Fbm^\times$ such that $\Tr_1^m(\beta D_\gamma F(x))=0$
           for all $x\in \Fbn$, where $F(x)=F_0(x)+\Tr_m^n(cx)$ and $c\in \Fbn$ is determined by
           the condition $\Tr(c\beta\gamma)=\Tr_1^m(\beta (F_0(\gamma)+F_0(0)))$.
            Choosing $\alpha\in \Fbn$ satisfying $\Tr(\alpha\gamma)=0$, the function
            $$
              F\circ H(x)=F(x)+\eps D_\gamma F(x) \quad (\,\,\text{with }\eps=\Tr(\alpha x) + \Tr_1^m(\beta F(x)) \,\,)
            $$
         has an algebraic degree $3$ and satisfies $F\circ H \eqccz F_0$ and $F\circ H \noteqea F_0$.
         %Moreover
         %$F\circ H$ is EA-inequivalent to any power function.

 \item[2.] When $p>2$: Suppose that $|\cW_{F_0}|\leq p^{n-1}$ and $p^{n-m}<\Delta_{F_0}\leq p^{n-3}$.
         Then there exist $\gamma\in\Fpn^\times$ and $\beta\in \Fpm^\times $ such that $\Tr_1^m(\beta D_\gamma F(x))=0$
           for all $x\in \Fpn$, where $F(x)=F_0(x)+\Tr_m^n(cx)$ and $c\in \Fpn$ is determined by
           the condition $\Tr(c\beta\gamma)=-\Tr_1^m(\beta (F_0(\gamma)-F_0(0)))$.
            Choosing $\alpha\in \Fpn$ satisfying $\Tr(\alpha\gamma)=-2$, the function
            $$
              F\circ H(x)=F(x)+\eps D_\gamma F(x) +F_{\text{\tiny DO}}(\gamma)(\eps^2-\eps)
                \quad (\,\,\text{with }\eps=\Tr(\alpha x) +\Tr_1^m(\beta F(x)) \,\,)
            $$
         satisfies $F\circ H \eqccz F_0$ and $F\circ H \noteqea F_0$, where
         $d^\circ (F\circ H)=3$ if $F_{\text{\tiny DO}}(\gamma)=0$, and
         $d^\circ (F\circ H)=4$ if $F_{\text{\tiny DO}}(\gamma)\neq 0$.
         %Furthermore, when $F_{\text{\tiny DO}}(\gamma)=0$, then
         %$F\circ H$ is EA-inequivalent to any power function. Also, when $F_{\text{\tiny DO}}(\gamma)\neq 0$,
         %by assuming that $F_0$ is a DO-polynomial (one can always assume this up to EA-equivalence),
         %$F\circ H$ is EA-inequivalent to any power function if $\Tr(c\beta\gamma)=0$.
 \end{enumerate}

\end{theorem}

\noindent The above theorem implies that, for most of quadratic $(n,m)$-functions $F$, there exists a function $F'$
such that
 $F'\eqccz F$ but $F'\noteqea F$. In other words, the notion of CCZ-equivalence is strictly more general than the
 notion of EA-equivalence in the following sense.

\begin{corollary}
Let $p$ be an arbitrary prime. Let $n\geq 4$ and $m$ is a divisor of $n$ such that $4 \leq m <n$. Let $F$ be a
quadratic $(n,m)$-function on $\Fpn$. Suppose that
 $$
 |\cW_{F}|\leq p^{n-1} \text{ and }  p^{n-m}<\Delta_{F}\leq p^{n-3}.
 $$
 Then the CCZ-equivalence class of $F$ is strictly larger than the EA-equivalence class of $F$. That is, one has $\clsEA(F) \subsetneqq
 \clsCCZ(F)$.
\end{corollary}

\begin{remark}
It should be mentioned that two examples of $(n,m)$-functions $F$ with $m|n$ satisfying $\clsEA(F) \subsetneqq
 \clsCCZ(F)$ are already known. One is the binary case $F(x)=\Tr_m^n(x^3)$ in \cite{BC10eqccz}, and the other is a
 $p$-ary ($p>2$) case $F(x)=\Tr_m^n(x^2-x^{p+1})$ in \cite{BH11}. Moreover, for a given $(n,m)$-function $F$ satisfying
$\clsEA(F) \subsetneqq \clsCCZ(F)$ and for any $k\geq 1$, it is shown in \cite{BC10eqccz} that one can construct
$(n,m+k)$-function $F'$ satisfying $\clsEA(F') \subsetneqq \clsCCZ(F')$. See also \cite{PZ13} for a generalization of
these methods over finite abelian groups.
\end{remark}

\noindent  Please note that, when  $m=n$, Theorem \ref{mainthm2} is exactly same to Theorem \ref{mainthm}. When $m$ is
a proper divisor of $m$ (i.e., $m<n$), then one needs a stronger assumption of the lower bound of $\Delta_F$ in Theorem
\ref{mainthm2},
 that is, $p^{n-m} < \Delta_F$. Moreover, the case $\Delta_F=p^{n-m}$ happens if and only if $F$ is a bent (or perfect non-linear)
  $(n,m)$-function.
 Recall that, in Section \ref{preli}, we defined that $(n,1)$-function $f$ on $\Fpn$ is bent if
 $|\cW_f|=p^{\frac{n}{2}}$. In a similar manner, we define that $(n,m)$-function $F$ on $\Fpn$ is bent if
 every component function $\Tr_1^m(bF) \,\,(b\in \Fpm^\times)$ is bent. Then the following equivalent conditions are well-known
  (See for example \cite{Nyb92}.):
  \begin{enumerate}
  \item[--] $F: \Fpn \rightarrow \Fpm$ is bent.
  \item[--] $D_\gamma \Tr_1^m(\beta F)$ is balanced for every $\beta \in \Fpm^\times$ and $\gamma \in\Fpn^\times$.
  \item[--] $D_\gamma F$ is balanced for every $\gamma \in \Fpn^\times$.
  \item[--] $\Delta_F=p^{n-m}$.
  \end{enumerate}

  \noindent Therefore, \emph{explicit examples} of $(n,m)$-function $F$ satisfying the assumptions of Theorem \ref{mainthm2}
  can be constructed by finding non-bent $(n,m)$-function $F$ satisfying $\Delta_F \leq p^{n-3}$. Since every
  $(n,m)$-function is of the form $\Tr_m^n(G)$ for some $(n,n)$-function $G$, we may construct explicit examples of
  $(n,m)$-functions satisfying Theorem \ref{mainthm2} using our previous results in Section \ref{sec-degree} as follows.

\begin{theorem}\label{thm-mn}
Let $F_0(x)$ be a quadratic $(n,n)$-function on $\Fpn$ satisfying the conditions of Theorem \ref{mainthm} with
corresponding $\alpha,\beta, \gamma, c$ in $\Fpn$ and $F(x)=F_0(x)+cx$. Letting $\Delta_F =p^s$, assume further that
$0<s\leq m-3$ (i.e., $p^{n-m}<p^{n-m+s}\leq p^{n-3}$) where $m|n$ with $m \geq 4$. Then the $(n,m)$-function
$$F_0'(x)=\Tr_m^n(\beta F(x))$$  satisfies the
conditions of Theorem \ref{mainthm2} with $\alpha'=\alpha, \beta'=1, \gamma'=\gamma$ and $c'=0$. Therefore, using
 $$\eps'=\Tr(\alpha' x) + \Tr_1^m(\beta' F'(x))=\Tr(\alpha x) + \Tr_1^m( \Tr_m^n(\beta F(x)))=\Tr(\alpha x +\beta F(x))=\eps,$$
 one has the followings.

 \begin{enumerate}
 \item[1.] When $p=2$:
  \begin{align*}
   F'\circ H'&=F'+\eps' D_{\gamma'} F'=F'+\eps D_{\gamma} F' \\
     &=\Tr_m^n(\beta F)+\eps  \Tr_m^n(\beta D_\gamma F)
   \end{align*}
   is
  CCZ-equivalent to $F'=\Tr_m^n(\beta F)$ but not EA-equivalent to each other.

 \item[2.] When $p>2$:
   \begin{align*}
    F'\circ H'&=F'+\eps' D_{\gamma'} F' +{F}_{\text{\tiny DO}}'(\gamma')({\eps'}^2-\eps')\\
      &=\Tr_m^n(\beta F)+\eps \Tr_m^n(\beta D_{\gamma}F) +{\Tr_m^n(\beta F_{\text{\tiny DO}})}(\gamma)({\eps}^2-\eps)
    \end{align*}
    is
  CCZ-equivalent to $F'=\Tr_m^n(\beta F)$ but not EA-equivalent to each other.
 \end{enumerate}
\end{theorem}

\noindent
\begin{proof}
Since $c'=0$, one has $F'(x)=F_0'(x)=\Tr_m^n(\beta F(x))$ and
$$
\Tr_1^m(\beta' D_{\gamma'}F'(x))=\Tr_1^m(D_{\gamma}\Tr_m^n(\beta F(x)))
 =\Tr_1^m(\Tr_m^n(\beta D_\gamma F(x)))=\Tr(\beta D_\gamma F(x))=0,
$$
which says that $\gamma'=\gamma$ is a $0$-linear structure of $\Tr_1^m(F')$.
 Also, for $a\in \Fpn$ and $b\in \Fpm$, one has
 \begin{align*}
 \cW_{F'}(a,b)&=\sum_{x\in \F_{p^n}}\xi^{\Tr_1^m(bF'(x))-\Tr(ax)}=\sum_{x\in
\F_{p^n}}\xi^{\Tr_1^m(b\Tr_m^n(\beta F(x)))-\Tr(ax)} \\
 &=\sum_{x\in \F_{p^n}}\xi^{\Tr_1^m(\Tr_m^n(b\beta F(x)))-\Tr(ax)}
 =\sum_{x\in \F_{p^n}}\xi^{\Tr (b\beta F(x))-\Tr(ax)}= \cW_{F}(a,b\beta).
 \end{align*}
 Therefore one gets $|\cW_{F'}|=|\cW_{F}|< 2^n$ for a binary case and $|\cW_{F'}|=|\cW_{F}|\leq p^{n-1}$ for odd
characteristic case. Note that $F'=\Tr_m^n(\beta F)$ is not a bent function because $\Tr_1^m (F')=\Tr(\beta F)$ is not
a bent function (since $\gamma$ is a $0$-linear structure of $\Tr(\beta F)$). Thus one has the lower bound
$p^{n-m}<\Delta_{F'} $. For the upper bound, one gets $\Delta_{F'}=\underset{\delta\in \Fpn^\times}{\max} |\Ker
(\Tr_m^n \circ \beta L_\delta)| \leq p^{n-m+s}\leq p^{n-3}$ where $L_\delta(x)=F(x+\delta)-F(x)-(F(\delta)-F(0))$.
\end{proof}

 \bigskip
 \noindent An obvious implication of the above theorem for the case of APN function over $\Fbn$ is the
 folowing.
\begin{corollary}\label{coro-mn}
For every quadratic APN function $F_0(x)$ on $\Fbn$ $(n\geq 4)$ with $F(x)=F_0(x)+cx$ satisfying $\Tr(\beta D_\gamma
F(x))=0$ for all $x$, one has $\Delta_{F'}=2^{n-m+1}$ with $F'=\Tr_m^n(\beta F)$. Thus the condition $2^{n-m}
<\Delta_{F'}\leq 2^{n-3}$ is satisfied for every positive divisor $m$ of $n$ with $m \geq 4$, and the following
function
 $$
   F'+\eps D_\gamma F' =\Tr_m^n(\beta F)+\eps \Tr_m^n(\beta D_\gamma F) \qquad  (\eps=\Tr(\alpha x+\beta F(x)))
 $$
  is of algebraic degree $3$ and is CCZ-equivalent to $\Tr_m^n(\beta F)$ but not EA-equivalent to each other.
\end{corollary}

\begin{remark}
When $m=n$, one has $F'=\beta F$ and $F'+\eps D_\gamma F'=\beta (F+\eps D_\gamma F)$ so that one gets the same result
of the case of $(n,n)$-functions in Theorem \ref{mainthm}. Also when $m=\frac{n}{2}$ with even $n$, one gets
$(2m,m)$-function $F'$ such that $F'$ has a CCZ-equivalent but not EA-equivalent function $F'+\eps D_\gamma F'$. That
is, examples of $(2m,m)$ function $F'$ satisfying $\clsEA (F') \subsetneqq \clsCCZ (F')$.
\end{remark}

\noindent We present more explicit examples related to the equivalent classes of $(n,m)$-functions arising from Gold
functions over $\Fbn$.

\begin{corollary}{$($Binary $(n,m)$-functions arising from Gold function$)$}\label{coro-goldnm}
 Let $n\geq 4$ and $m$ be a divisor of $n$ with $m\geq 4$. Let $i\geq 1$ be an integer such that
 $\gcd(i,n)=1$ (i.e., $x^{2^i+1}$ is APN). Then,
 $$
  \Tr_m^n(x^{2^i+1}) +\Tr_m^n(x^{2^i}+x) \Tr(x^{2^i+1}+x)
 $$
 is of algebraic degree \it{three}, and is CCZ-equivalent to  $\Tr_m^n(x^{2^i+1})$ but not EA-equivalent to each other.
 \end{corollary}

 \noindent
 \begin{proof}
 This is an immediate consequence of Theorem \ref{thm-mn} or Corollary \ref{coro-mn} saying that
 $$
   \Tr_m^n(F(x)) \noteqea  \Tr_m^n(F(x)) +\Tr_m^n(D_1 F(x)) \Tr(x^{2^i+1}+x) \eqccz \Tr_m^n(F(x))
 $$
 with $F(x)= x^{2^i+1}+x$ and $\gamma=1, \beta=1, \alpha=0, c=1$ so that $\eps=\Tr(x^{2^i+1}+x)$,
  which is analogous to  Remark \ref{goldremark} of Section
 \ref{subsec-goldbinary}.
 \end{proof}

\begin{remark}
From Corollary \ref{coro-goldnm}, consider the case that $n$ is even and $m=\frac{n}{2}$ ($ m\geq 4$ and
$\gcd(i,2m)=1$). Then  one gets
 \begin{align*}
  \Tr_m^{2m}(x^{2^i+1})= x^{2^i+1}+x^{(2^i+1)2^m} \text{ {\rm and} }\,  \Tr_m^{2m}(x^{2^i}+x)=x+x^{2^i}+x^{2^m}+x^{2^{i+m}}
 \end{align*}
 so that the following cubic $(2m,m)$-function
  $$
     x^{2^i+1}+x^{(2^i+1)2^m} +\left(x+x^{2^i}+x^{2^m}+x^{2^{i+m}}\right)\Tr\left(x^{2^i+1}+x\right)
  $$
  is CCZ-equivalent to $x^{2^i+1}+x^{(2^i+1)2^m}$ but not EA-equivalent to each other. We think that this is an
  interesting example in the sense that the greatest $1 \leq m \leq n $ such that there exists a binary bent $(n,m)$-function is
  the case $m=\frac{n}{2}$ with $n=$even (See \cite{Nyb92}.). Moreover, it is shown in \cite{BC09eqccz} that CCZ-equivalence coincides with
  EA-equivalence for the case of bent functions. Summing up these together, we have the following remark.
\end{remark}

\begin{remark}
Let $n\geq 8$ be even and let $F$ be an $(n,m)$-function on $\F_2^n$. Then, $m=\frac{n}{2}$ is the greatest possible
integer $m$ such that there exist two $(n,m)$-functions $F$ and $F'$ such that $\clsEA(F)=\clsCCZ(F)$ and
$\clsEA(F')\subsetneqq \clsCCZ(F')$. Examples of such $F$ are bent functions, and examples of $F'$ are non-bent
functions satisfying $|\cW_{F'}|<2^n$ and $\Delta_{F'}\leq 2^{n-3}$.
\end{remark}

%\bigskip
%\bigskip
%???

%\begin{remark}
%Note that $\Tr_m^n(x^{2^i}+x)= \Tr_m^n( D_1 F(x))=D_1 \Tr_m^n(F(x))$. Therefore the assumption $\gcd(i,n)=1$ implies
%that $\Tr_m^n(x^{2^i}+x)$ is a nonzero function (i.e., $x^{2^i}$ and $x$ are not conjugate over $\Fbm$).
%\end{remark}

%\begin{remark}
%Please verify the above theorem for $(n,m)=(8,4), (10,5), (12,4), (12,6)$ if possible.
%\end{remark}

%\bigskip
%\bigskip
%???

%\bigskip
%\bigskip

\noindent For non-binary cases, examples of $(n,m)$-functions on $\Fpn$ satisfying Theorem \ref{mainthm2} and
\ref{thm-mn} can also be given in the same manner. We will give such examples arising from Gold function
$G(x)=x^{p^i}+1$ on $\Fpn$ with {\it even} $n$ and $\gcd(i,n)=1$. In this case, from Section \ref{sec-paryGold}, recall
that $|\cW_G|=p^{\frac{n}{2}+1}$ and $\Delta_G=p$. Therefore, with the same parameters $c=0, \gamma=1$ and
$\alpha,\beta \in \Fpn$
 as in Remark \ref{rmkalpbet} and Corollary \ref{coro-quartic2}, one has

 \begin{corollary}{$($Non-binary $(n,m)$-functions arising from Gold function$)$}\label{coro-parygold}
 Let $p$ be an odd prime.  Let $n\geq 4$ be even and $m\geq 4$ be a divisor of $n$.
 Let $\alpha\in \Fpn$ and $\beta\in \F_{p^2}^\times$ such that $\Tr(\alpha)=-2$ and $\beta^p+\beta=0$.
 Let $\eps=\Tr(\alpha x +\beta x^{p^i+1})$ and $\gcd(i,n)=1$.
  Then the following
 $(n,m)$-function on $\Fpn$,
   \begin{align}
      \Tr_m^n(\beta x^{p^i+1})+\eps \Tr_m^n \left( \beta (x^{p^i}+x+1) \right)+
       \Tr_m^n(\beta)(\eps^2-\eps),  \label{nm-pary}
   \end{align}
   is of algebraic degree $3$ or $4$, and is CCZ-equivalent to $\Tr_m^n(\beta x^{p^i+1})$ but not EA-equivalent to each
   other.
 \end{corollary}
 \begin{proof}
   Letting $F'(x)=\Tr_m^n(\beta x^{p^i+1})$, Theorem \ref{thm-mn} implies that
   $$
     F'(x) \eqccz F'(x)+\eps D_\gamma F'(x) + F'_{\text{\tiny DO}}(\gamma) (\eps^2-\eps) \noteqea  F'(x),
   $$
   where
    \begin{align*}
    D_\gamma F'(x) &=D_1 \Tr_m^n(\beta x^{p^i+1})=\Tr_m^n( \beta D_1 x^{p^i+1}) \\
      & = \Tr_m^n(\beta( (x+1)^{p^i+1}-x^{p^i+1} ) )=\Tr_m^n(\beta (x^{p^i}+x+1)), \text{ and}\\
       F'_{\text{\tiny DO}}(\gamma)&=F'(\gamma)=F'(1)=\Tr_m^n(\beta).
    \end{align*}
    It is also straightforward to see $\Delta_{F'}=p^{n-m+1}$.
 \end{proof}

\begin{remark}
 The $(n,m)$-function in Corollary \ref{coro-parygold} is of algebraic degree three
  if $\Tr_m^n(\beta)=0$, and four if $\Tr_m^n(\beta)\neq 0$. Since $n$ is even and $\beta \in \F_{p^2}$,
  one has  $\Tr_m^n(\beta)=0$ if $m$ is odd, and  $\Tr_m^n(\beta)=\frac{n}{m} \beta$ if $m$ is even.
  Therefore, the function in \eqref{nm-pary} is cubic if $m$ is odd or $\frac{n}{m} \equiv 0 \pmod{p}$, and is quartic
  otherwise.
\end{remark}

\begin{corollary}\label{coro-3arygold}
 Let $p=3$.  Let $n\geq 4$ be even and $m\geq 4$ be a divisor of $n$.
 Let $\alpha\in \F_{3^n}$ and $\beta\in \F_{3^2}^\times$ such that $\Tr(\alpha)=1$ and $\beta^2=-1$.
 Let $\eps=\Tr(\alpha x +\beta x^{3^i+1})$ and $\gcd(i,n)=1$.
  Then the following
 $(n,m)$-function on $\F_{3^n}$,
   $$
      \Tr_m^n(\beta x^{3^i+1})+\eps \Tr_m^n \left( \beta (x^{3^i}+x+1) \right)+
       \Tr_m^n(\beta)(\eps^2-\eps),
   $$
   is of algebraic degree $3$ if $m\equiv 1\pmod{2}$ or $\frac{n}{m} \equiv 0 \pmod{3}$, and is of degree $4$ otherwise,
   and is CCZ-equivalent to $\Tr_m^n(\beta x^{3^i+1})$ but not EA-equivalent to each
   other.
 \end{corollary}

\section{Conclusions and Future Works}

For almost all quadratic $(n,n)$-functions $F_0$ over $\Fpn$ with $n\geq 4$, we constructed $F'$ satisfying $F'\eqccz
F_0$ but $F'\noteqea F_0$. For a binary case, assuming $\NL(F_0)\neq 0$ and $\Delta_{F_0} \leq 2^{n-3}$, the function
$$
F'=F+\eps D_\gamma F \quad (\text{with } F(x)=F_0(x)+cx \text{ and } \eps=\Tr(\alpha x +\beta F(x)))
$$
satisfies $F_0\noteqea F' \eqccz F_0$ for many suitably chosen $\alpha, \beta, \gamma $ and $c$ in $\Fbn$, where the
number of possible choices of $\{\alpha,\beta, \gamma, c\}$ is at least $2^{2n-2}(2^n-1)$ and exactly $2^{2n-2}(2^n-1)$
if $F_0$ is APN. For a field $\Fpn$ with $p>2$, assuming $ |\cW_{F_0}|\leq p^{n-1}$ and $1<\Delta_{F_0}\leq p^{n-3}$,
the function
$$
F'=F+\eps D_\gamma F +F_{\text{\tiny DO}}(\gamma)(\eps^2-\eps) \quad (\text{with } F(x)=F_0(x)+cx \text{ and }
\eps=\Tr(\alpha x +\beta F(x)))
$$
satisfies $F_0\noteqea F' \eqccz F_0$ for many suitably chosen $\alpha, \beta, \gamma $ and $c$ in $\Fbn$, where the
number of possible choices of $\{\alpha,\beta, \gamma, c\}$ is at least $p^{2n-2}(p-1)$.

For arbitrary characteristic $p$, one can choose $\alpha,\beta, \gamma$ and $c$ (in the order of $\gamma\mapsto
\{\beta,\alpha\} \mapsto c$) satisfying
 \begin{align*}
\Tr(\beta D_\gamma F_0(x))&=\Tr(\beta D_\gamma F_0(0)) \text{ for all } x\in \Fpn,  \,\,\Tr(\alpha\gamma)=-2, \\
\text{\emph{and} }\,\, \Tr(c\beta\gamma)&=-\Tr(\beta D_\gamma F_0(0)).
 \end{align*}

\noindent For a fixed $F_0$, such  construction generates many $F'$ because $F'$ is depending on the choices of
$\alpha, \beta,\gamma$ and $c$. Though all of such $F'$ are CCZ-equivalent to each other, EA-equivalence or
(EA-inequivalence) among them are unclear at this moment. Moreover such $F'$ are no longer plateaued, and the
corresponding the Walsh spectrum needs to be investigated further. Also, it may be possible to apply our technique to
general $(n,m)$-functions, where $m$ is not necessarily a divisor of $n$. Our technique (especially in Proposition
\ref{key-prop}) heavily relies on the equivalence between the linear structure and the linear translator over the prime
field $\F_p$, where the linear map $\Tr_1^n$ has its image in $\F_p$. That is the reason why we could not extended the
ideas in \cite{BCP06, BCL09FFTA}, where they also used the trace maps over the intermediate fields such as $\Tr_3^n$ or
$\Tr_d^n$ with $d|n$. Therefore, extending Proposition \ref{key-prop} to the trace map over the intermediate fields is
also desirable.

\end{document}